\documentclass[10pt]{article}
\usepackage[T1]{fontenc}
\usepackage{amsmath,amsfonts,amsthm,mathrsfs,amssymb}
\usepackage{cite}
\usepackage{multirow}
\usepackage{graphicx,float}
\usepackage{subfigure}
\usepackage{placeins}
\usepackage{color}
\usepackage{indentfirst}
\usepackage[bookmarks=true]{hyperref}
\numberwithin{equation}{section}
\newcommand{\R}{{\mathbb R}}

\newcommand{\be}{\begin{eqnarray}}
\newcommand{\ben}{\begin{eqnarray*}}
\newcommand{\en}{\end{eqnarray}}
\newcommand{\enn}{\end{eqnarray*}}

\newcommand{\pa}{\partial}

\newcommand{\ov}{\overline}
\newcommand{\curl}{{\rm curl\,}}

\newcommand{\divv}{{\rm div\,}}

\newcommand{\G}{\Gamma}

\newtheorem{theorem}{Theorem}[section]
\newtheorem{lemma}[theorem]{Lemma}

\newtheorem{remark}[theorem]{Remark}

\definecolor{rot}{rgb}{1.000,0.000,0.000}
\definecolor{rot1}{rgb}{0.000,0.000,0.000}
\begin{document}
\renewcommand{\theequation}{\arabic{section}.\arabic{equation}}
\begin{titlepage}
\title{Regularized integral equation methods for \\ elastic scattering problems in three dimensions}

\author{Oscar P. Bruno\thanks{Department of Computing \& Mathematical
    Sciences, California Institute of Technology, 1200 East California
    Blvd., CA 91125, United States. Email:{\tt obruno@caltech.edu}}\;
  and Tao Yin\thanks{Department of Computing \& Mathematical Sciences,
    California Institute of Technology, 1200 East California Blvd., CA
    91125, United States. Email:{\tt taoyin89@caltech.edu}}}
\end{titlepage}
\maketitle

\begin{abstract}
  This paper presents novel methodologies for the numerical simulation
  of scattering of elastic waves by both closed and open surfaces in
  three-dimensional space. The proposed approach utilizes new integral
  formulations as well as an extension to the elastic context of the
  efficient high-order singular-integration methods~\cite{BG18}
  introduced recently for the acoustic case. In order to obtain
  formulations leading to iterative solvers (GMRES) which converge in
  small numbers of iterations we investigate, theoretically and
  computationally, the character of the spectra of various operators
  associated with the elastic-wave Calder\'on relation---including
  some of their possible compositions and combinations. In particular,
  by relying on the fact that the eigenvalues of the composite
  operator $NS$ are bounded away from zero and infinity, new
  uniquely-solvable, low-GMRES-iteration integral formulation for the
  closed-surface case are presented. The introduction of corresponding
  low-GMRES-iteration equations for the open-surface equations
  additionally requires, for both spectral quality as well as accuracy
  and efficiency, use of weighted versions of the classical integral
  operators to match the singularity of the unknown density at
  edges. Several numerical examples demonstrate the accuracy and
  efficiency of the proposed methodology.

{\bf Keywords:} Elastic waves, combined field integral equations, Calder\'on relation, hyper-singular operator, high-order methods
\end{abstract}

\section{Introduction}
\label{sec:1}

This paper introduces efficient high-order integral solvers for
three-dimensional (3D) problems of elastic scattering by surfaces,
including both {\it closed surfaces} and {\it infinitely thin open
  surfaces}. These are problems of significant importance in many
application fields in science and engineering, including geophysics,
seismology, non-destructive testing, energy and material science,
among many others.  Unlike the finite-element and finite-difference
approximations of the associated partial differential
equations~\cite{BHSY}, which require 3D (volumetric) discretizations
and use of appropriate absorbing boundary conditions, the boundary
integral methods only require discretization of the two-dimensional
(2D) domain boundaries~\cite{CK98,HW08,KGBB79,N01} and they
automatically enforce the radiation condition at infinity.  The
significant benefits inherent in the reduced dimensionality of the
boundary-integral methods can be fully realized, in spite of the dense
matrices they entail, provided adequate acceleration techniques are
used for the associated matrix-vector products (see
e.g.~\cite{BK01,BET12,L09} and references therein) together with
Krylov-subspace linear algebra solver like GMRES. In all, the BIE
method has lead to fast and high-order algorithms that, for problems
beyond a small number of wavelengths in size, can outperform their
volumetric domain discretization counterparts to very significant
extents.

For definiteness, this contribution focuses on the elastic Neumann
(traction) boundary-value problem, whose treatment by means of
boundary integral methods has been found quite challenging, but the
proposed methods extend directly to the somewhat less challenging
elastic problems with Dirichlet (displacement) boundary
conditions. For the Neumann problem of scattering by closed-surfaces
the proposed method represents the elastic-field on the basis of a
combination of single-layer and double-layer potentials~\cite{CK98},
which ensures the validity of the critical property of unique
solvability; the resulting integral equation includes contributions
from the tractions of both the single-layer and double-layer
potentials---which result in strongly singular and hyper-singular
kernels, respectively, and which, unlike the single layer operator
(whose kernel is weakly singular), are only defined in the sense of
Cauchy principle value and Hadamard finite part~\cite{HW08},
respectively.  For the problem of scattering by open-surfaces, in
turn, a representation leading to a hyper-singular integral operator
is used~\cite{AD95,CKM00}. In both cases we propose an efficient
high-order singular-integration method that extends the
``rectangular-polar'' methodology~\cite{BG18} introduced recently for
the acoustic case, and which, as demonstrated below in this paper, can
efficiently produce solutions of very high accuracy.

The presence of the elastic hyper-singular operator in the
integral-equation formulations presents difficulties concerning
spectral character and accurate operator evaluation, both of which
arise from the highly singular character of the associated integral
kernel. Indeed, as it is well known, the eigenvalues of the
hyper-singular operators accumulate at infinity and, hence, the
solution of these integral equations by means of the GMRES solver
often requires large numbers of iterations for convergence---and thus,
large computing costs, specially for 3D problems. On the other hand,
the very evaluation of the associated hyper-singular integrals for
both open- and closed-surface problems, that must be interpreted in
the sense of Hadamard finite part, has also remained a significant
challenge~\cite{HW08,BLR14}. When combined with the iterative
linear-algebra solver GMRES, the proposed combination of a spectrally
regularized formulation and novel and effective high-order singular
quadratures gives rise to efficient and highly accurate solvers for
the elastic-wave problems at hand. We suggest that the use of the
aforementioned acceleration techniques, which can directly be applied
in conjunction with the formulation and singular-quadrature methods
presented in this paper, would lead to accurate and efficient solvers
for high-frequency elastic-scattering problems as well.

A number of strategies have been developed, in the context of acoustic
and electromagnetic scattering, for reduction of the number of GMRES
iterations required for convergence to a given accuracy. Unlike the
algebraic preconditioners \cite{BT98,CDGS05} and formulations based on
pseudoinverses \cite{AD05,AD07}, the novel methodologies proposed in
\cite{BET12} rely on the acoustic Calder\'on relation and only require
use of a regularizing operator of a form similar to a single-layer
operator, leading to regularized integral equations that are of the
desired second-kind Fredholm type. In both cases the regularization
technique preserves the unique solvability properties of the classical
(unregularized) integral equations upon which they are based.

The extension of these methodologies to elastic scattering problems
presents certain challenges. At a basic level, elastic-wave Calder\'on
formulas have not been studied in detail for either closed-surface or
open-surface cases---possibly on account of the fact that, in contrast
with the acoustic wave case, the classical double-layer operator $K$
and its adjoint $K'$ (which play important roles in the Calder\'on
relations) are not compact in the elastic case~\cite{AJKKY,AKM}. The
2D elastic Calder\'on formula for the closed-surface case was
investigated recently~\cite{BXY191}. On the basis of the polynomial
compactness of the operators $K$ and $K'$ it was shown that the
composition $NS$ of the single-layer and hyper-singular integral
operators can be expressed as the sum of a multiple of the identity
operator and a compact operator. The closed-surface 2D analysis does
not directly translate to the 3D context in view of certain
differences in the detailed character of the polynomial compactness of
the operators $K$ and $K'$ in the 3D~\cite{AKM} and 2D~\cite{AJKKY}
cases, but, as shown in Section~\ref{sec:3}, the eigenvalues of the
composition $NS$ in 3D are bounded away from zero and
infinity. Analyses based on principal symbols, such as the one
presented in~\cite{DL15} for the on-surface radiation-condition
regularization method in the Dirichlet case, could conceivably be
applied to study the spectral regularity of the three-dimensional
elastic operator $NS$, but such approaches have not as yet been
pursued. We are not aware of previous applications of spectral
regularization methods to integral equations for the elastic Neumann
problem.

Elastic versions of the Calder\'on formulas for open surfaces are not
known at present.  In view of the acoustic open-surface Calder\"on
relations~\cite{BL12,BL13,LB15} and the related study~\cite{BXY191}
for the 2D open-arc elastic case, we consider ``weighted'' versions
$S_w$ and $N_w$ of the single-layer and hyper-singular operators
which, like those considered previously for 2D elastic and 2D and 3D
acoustic open surface problems, extract the solutions' edge
singularity explicitly. In view of these contributions and the
spectral properties, established in the present paper for the 3D
closed-surface elastic case, we additionally consider a formulation of
the 3D open-surface elastic problem in terms of the composition
$N_wS_w$. The benefits of this approach are two-fold:
high-order accuracy (that is achieved in our implementations by means
of the aforementioned rectangular-polar quadrature method) and
well-behaved iterative linear algebra. Our numerical tests suggest
that the eigenvalues of $N_wS_w$ are at least bounded away from
infinity, and although they appear to approach the origin
(Fig.\ref{eig3}), the $N_wS_w$ formulation leads, as desired, to
significant reductions in the number of GMRES iterations required for
convergence to a given residual tolerance over those required by the
operator $N_w$.

As indicated above, our implementations rely on the Chebyshev-based
rectangular-polar discretization methodology developed
recently~\cite{BG18}---which can be readily applied in conjunction
with geometry descriptions given by a set of arbitrary non-overlapping
logically-quadrilateral patches, and which, therefore, makes the
algorithm particularly well suited for treatment of complex
CAD-generated geometries. The algorithms additionally rely on use of
expressions, presented in~\cite{BXY19,L14,YHX} for the closed-surface
case, that present 3D elastic strongly-singular and hyper-singular
operators as compositions of weakly singular integrals and
tangential-derivative operators; the corresponding expressions for the
weighted operators we use in the open-surface case are presented in
Lemma~\ref{RegOpen} below.  The application of the Chebyshev-based
rectangular-polar solver for the evaluation of the weakly singular
integrals gives rise to high accuracy and efficiency. Thanks to the
use of Cartesian-product Chebyshev discretizations, further, the
needed tangential differentiations can easily be effected via
differentiation of corresponding truncated Chebyshev expansions.

This paper is organized as follows. After preliminaries and notations
are laid down in Section \ref{sec:2.1}, Sections~\ref{sec:2.2}
and~\ref{sec:2.3} introduce the classical integral equations for the
closed-surface and open-surface problems under consideration,
respectively.  Section~\ref{sec:3.1} investigates the spectral
properties of the elastic integral operators and the 3D Calder\'on
relation. The new regularized integral equations for closed and open
surfaces are derived in Sections~\ref{sec:3.2} and~\ref{sec:3.3},
respectively. Exact re-expressed formulations for the
strongly-singular and hyper-singular operators are presented in
Section~\ref{sec:3.4}. The high order discretization method we use for
numerical evaluation of the elastic integral operators are briefly
described in Section~\ref{sec:4}. The numerical examples presented in
Section~\ref{sec:5}, finally, demonstrate the high-accuracy and
high-order of convergence enjoyed by the proposed approach, as well as
the reduced numbers of GMRES linear-algebra iterations required by the
proposed algorithms for convergence to a given residual tolerance.

\section{Elastic scattering problems and integral equations}
\label{sec:2}


\subsection{Preliminaries}
\label{sec:2.1}

We consider the problems of scattering of elastic waves by bounded
obstacles $\Omega$ whose smooth boundaries $\Gamma$ are either open or
closed surfaces---that is, they are two-dimensional sub-manifolds of
$\mathbb{R}^3$ with or without boundary, respectively. Noting that in
the open-surface case we have $\Omega=\Gamma$, for both the open- and
closed-surface cases the propagation domain will be denoted by
$D:=\R^3\backslash\ov{\Omega}$. We assume that $D$ is occupied by a
linear isotropic and homogeneous elastic medium characterized by the
Lam\'e constants $\lambda$ and $\mu$ (satisfying $\mu>0$,
$3\lambda+2\mu>0$) and the mass density $\rho>0$. As indicated in
Section~\ref{sec:1}, our derivations are restricted to the challenging
Neumann case, in which the boundary traction is
prescribed. Suppressing the time-harmonic dependence $e^{-i\omega t}$
in which $\omega>0$ is the frequency, the displacement field
$u=(u^1,u^2,u^3)^\top$ in the solid (where $a^\top$ denotes
transposition of the vector or matrix $a$) can be modeled by the
following boundary value problem: Given the boundary data $F$ on
$\Gamma$, determine the scattered field $u$ satisfying \be
\label{Navier}
\Delta^{*}u +   \rho \omega^2u &=& 0\quad\mbox{in}\quad D,\\
\label{BoundCond}
T(\pa,\nu)u &=& F \quad \mbox{on}\quad \Gamma,
\en
and the Kupradze radiation condition (\cite{KGBB79})
\be
\label{RadiationCond}
\lim_{r \to \infty} r\left(\frac{\partial u_t}{\partial r}-ik_tu_t\right) = 0,\quad r=|x|,\quad t=p,s,
\en
uniformly with respect to all $\hat{x}=x/|x|\in\mathbb{S}^2:=\{x\in\R^3:|x|=1\}$. Here, $\Delta^{*}$ denotes the Lam\'e operator
\ben
\label{LameOper}
\Delta^* := \mu\,\mbox{div}\,\mbox{grad} + (\lambda +
\mu)\,\mbox{grad}\, \mbox{div}\,, \enn and $T(\pa,\nu)$ denotes the
boundary-traction operator
\ben
\label{stress-3D}
T(\pa,\nu)u:=2 \mu \, \partial_{\nu} u + \lambda \, \nu \, \divv u+\mu
\nu\times \curl u,\quad \nu=(\nu^1,\nu^2,\nu^3){^\top}, \enn where
$\nu$ is the outward unit normal to the boundary $\G$ and
$\partial_\nu:=\nu\cdot\mbox{grad}$ is the normal derivative. In
(\ref{RadiationCond}), $u_p$ and $u_s$ denote the compressional and
the shear waves, respectively, which are given by \ben
u_p=-\frac{1}{k_p^2}\,\mbox{grad}\,\mbox{div}\;u,\quad
u_s=\frac{1}{k_s^2}\,\mbox{curl}\,\mbox{curl}\;u, \enn where the wave
numbers $k_s,k_p$ are defined as \ben k_s := \omega/c_p,\quad k_p :=
\omega/c_s, \enn with \ben c_p:=\sqrt{\mu/\rho},\quad
c_s:=\sqrt{(\lambda+2\mu)/\rho}.  \enn If the scattered field is
induced by an incident displacement field $u^{inc}$ (e.g. a plane wave
or point source), then the boundary data is determined by
$F=-T(\pa,\nu)u^{inc}$.

The fundamental displacement tensor for the time-harmonic Navier
equation (\ref{Navier}) in $\R^3$ is given by
\be
\label{NavierFS}
E(x,y) = \frac{1}{\mu}\gamma_{k_s}(x,y) I + \frac{1}{\rho\omega^2}\nabla_x\nabla_x^\top \left[\gamma_{k_s}(x,y) - \gamma_{k_p}(x,y)\right],\quad x\ne y,
\en
where
\be
\label{HelmholtzFS}
\gamma_{k_t}(x,y) =\frac{\exp(ik_t|x-y|)}{4\pi|x-y|}, \quad x\ne
y,\quad t=p,s.
\en
is the fundamental solution of the Helmholtz equation in $\R^3$ with wave number $k_t$. Relying on the $p$-wave and $s$-wave Helmholtz
Green functions~\eqref{HelmholtzFS}, Sections~\ref{sec:2.2}
and~\ref{sec:2.3} present the classical indirect boundary integral
equations for the traction problems of scattering by closed and open
surfaces, respectively.

\subsection{Boundary integral equations I: closed-surface case}
\label{sec:2.2}
The classical indirect combined field integral equation formulation
assumes a representation of the scattered field given by a combination
of the form \cite{HW08}
\be
\label{sol1}
u(x)=(\mathcal{D}-i\eta\mathcal{S})(\varphi)(x), \quad x\in D, \en of
a single and a double-layer potential expressions $\mathcal{D}$ and
$\mathcal{S}$ given by
\be
\label{dl}
\mathcal{D}(\varphi)(x) &=& \int_{\Gamma} \left(T(\pa_y,\nu_y)E(x,y)\right)^\top \varphi(y)\,ds_y, \\
\label{sl}
\mathcal{S}(\varphi)(x) &=& \int_{\Gamma} E(x,y)\varphi(y)\,ds_y, \en
respectively. Operating with the traction operator on (\ref{sol1}),
taking the limit as $x\to\Gamma$ and using well-known jump
relations~\cite{HW08} to apply the boundary condition, the combined
field integral equation \be
\label{BIE1}
\left[i\eta \left(\frac{I}{2}-K'\right)+N\right](\varphi)= F
\quad\mbox{on}\quad \Gamma \en results. Here $I$ denotes the identity
operator, and $K':H^{s}(\Gamma)^3\rightarrow H^{s}(\Gamma)^3$ and
$N:H^{s}(\Gamma)^3\rightarrow H^{s-1}(\Gamma)^3$ denote the boundary
integral operators \be
\label{BIOK}
K'(\varphi)(x) &=& \int_\Gamma T(\pa_x,\nu_x)E(x,y) \sigma(y)\,ds_y, \quad x\in\Gamma,\quad\mbox{and}\\
\label{BION}
N(\varphi)(x) &=& \int_\Gamma T(\pa_x,\nu_x)
\left(T(\pa_y,\nu_y)E(x,y)\right)^\top u(y)\,ds_y, \quad x\in\Gamma,
\en
which are only defined in the sense of Cauchy principle value and
Hadamard finite part~\cite{HW08} respectively, in view of the strongly
singular and hyper-singular character of the corresponding kernels.
It can be shown that the integral equation (\ref{BIE1}) is uniquely
solvable for all real values of the frequency $\omega>0$ (see
e.g.~\cite{BXY}). But, as it is well known, the eigenvalues of the
hypersingular integral operator $N$ accumulate at infinity. As a
result, the solution of the integral equation (\ref{BIE1}) by means of
Krylov-subspace iterative solvers such as GMRES generally requires
large numbers of iterations.

\subsection{Boundary integral equations II: open-surface case}
\label{sec:2.3}

For the open-surface scattering problem the solution can be expressed
as a double-layer potential \be
\label{sol2}
u(x)=\mathcal{D}(\varphi)(x), \quad x\in D.  \en Operating with the
traction operator on (\ref{sol2}), taking the limit as $x\to\Gamma$
and applying the boundary condition, we obtain the boundary integral
equation \be
\label{BIE2}
N(\varphi)= F \quad\mbox{on}\quad \Gamma.  \en In addition to the
computational challenge inherent in the accurate integration of the
hypersingular kernel of the operator $N$, for the open-surface case
the solution $\varphi$ is itself singular at the edge of $\Gamma$, as is well
known---which leads to numerical methods of low order of accuracy
unless the algorithm appropriately accounts for the solution
singularity.

\section{Regularized boundary integral equations}
\label{sec:3}

In this section, we propose the regularized integral equations for the closed and
open surface scattering problems. Here, three types of ``regularization'' are employed:
\begin{itemize}
\item[I.] {\it Form of integral equation regularization}: deriving new integral equations (Sections~\ref{sec:3.2} and \ref{sec:3.3}) on a basis of the spectral properties of the composition of single layer operator and hyper-singular operator (Section~\ref{sec:3.1});
\item[II.] {\it Solutions' edge singularity regularization for open-surface cases}: introducing a weight function to extract the solutions' edge singularity explicitly (Section~\ref{sec:3.3});
\item[III.] {\it Strong-singularity and hyper-singularity regularization}: Re-expressing the strongly singular and hyper-singular integral operators into compositions of weakly-singular integral operators and differentiation operators in directions tangential to the surface (Section~\ref{sec:3.4}).
\end{itemize}

\subsection{Operator spectra}
\label{sec:3.1}

Seeking to derive  regularized boundary integral equations which
do not suffer from the difficulties described in the previous section,
we first study the spectra of the integral operators $K'$ and
the composite operator $N S$ where $S:H^s(\Gamma)^3 \rightarrow H^{s+1}(\Gamma)^3$ denotes the single-layer
operator
 \be
\label{BIOS}
S(\varphi)(x) &=& \int_\Gamma E(x,y) \sigma(y)\,ds_y, \quad
x\in\Gamma.
\en
Our study for wave-scattering problems relies on the
following result for zero-frequency (static) elasticity.
\begin{theorem}{\cite[Theorem 2.1, 2.2]{AKM}}
\label{theorem1}
Let $\Gamma$ denote a smooth closed surface in three-dimensional
space, let $K_0'$ denote the adjoint of the elastic double-layer
operator in the zero-frequency case $\omega=0$ and let
$P_3(t)=t(t^2-C_{\lambda,\mu}^2)$ where $C_{\lambda,\mu}$ is a
constant that depends on the Lam\'e parameters: \ben
C_{\lambda,\mu}=\frac{\mu}{2(\lambda+2\mu)}.  \enn Then
$P_3(K_0'):H^{-1/2}(\Gamma)^3\rightarrow H^{-1/2}(\Gamma)^3$ is compact. Furthermore, the spectrum of $K_0'$ consists
of three non-empty sequences of eigenvalues which converge to 0,
$C_{\lambda,\mu}$ and $-C_{\lambda,\mu}$, respectively.
\end{theorem}

\begin{figure}[htb]
\centering
\begin{tabular}{cc}
\includegraphics[scale=0.25]{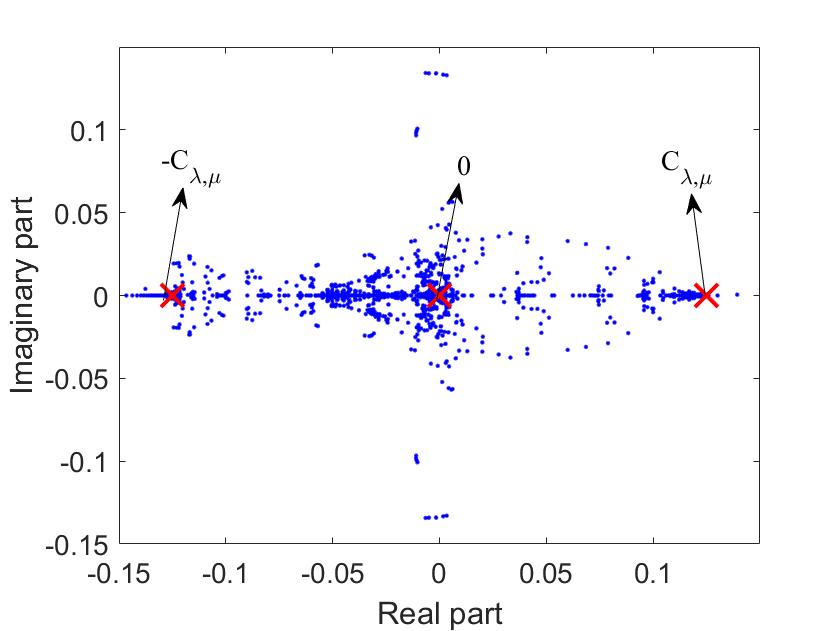} &
\includegraphics[scale=0.25]{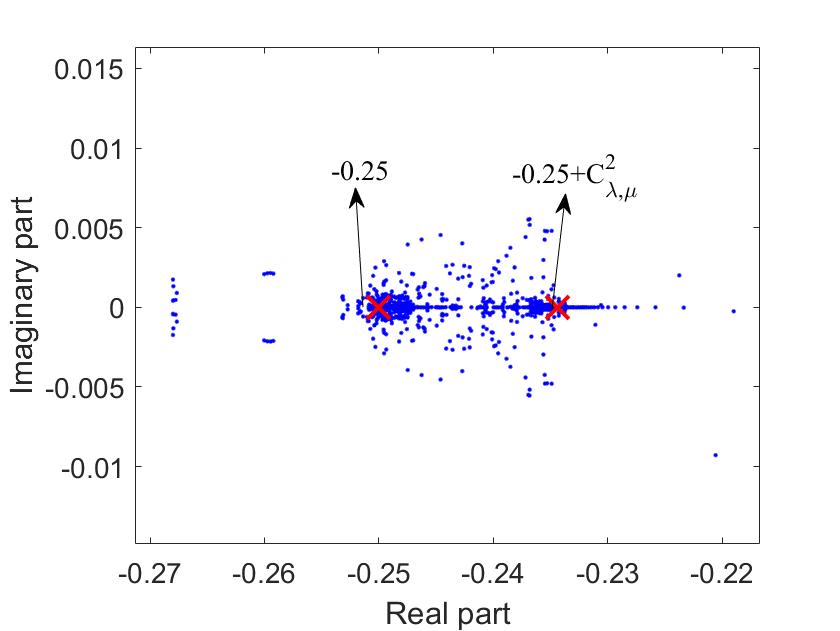} \\
(a) $K'$ & (b) $N S$ \\
\end{tabular}
\caption{Eigenvalue distributions for the integral operators $K'$ and
  $NS$. }
\label{eig1}
\end{figure}

Using this result we can explicitly obtain the accumulation points of
the eigenvalues of $K'$. Indeed, since $K'-K_0'$ has a weakly-singular
kernel it follows that $K'-K_0'$ is a compact operator, and we obtain
\ben P_3(K')=P_3(K_0')+K_c,\quad
K_c=K'(K'-K_0')(K'+K_0')+(K'-K_0')(K_0'^2-C_{\lambda,\mu}^2I), \enn
where $K_c$ is a compact operator. Therefore, the spectrum of $K'$
also consists of three sequences of eigenvalues which converge to 0,
$C_{\lambda,\mu}$ and $-C_{\lambda,\mu}$, respectively. In view of the
Calder\'on relation~\cite{HW08} \be
\label{caldron}
N S=-\frac{I}{4}+K'^2, \en together with the inequalities
$0 < C_{\lambda,\mu}<3/8$ (which result easily from the condition
$\lambda+2/3\mu>0$) we conclude that the eigenvalues of the composite
operator $N S$, which plays an essential role in the regularized
integral equations proposed in the following section, are bounded away
from zero and infinity.

To visualize the significance of these results we consider the
integral operators $K'$ and $N S$ associated with the problem of scattering by a unit
ball, and we choose $\lambda=2$, $\mu=1$, $\rho=1$,
$\omega=\pi$, from which we obtain $C_{\lambda,\mu}=0.125$. Letting
$N_\mathrm{DOF}$ denote the number of degrees of freedom used for
operator discretization, the eigenvalue distributions for the various
operators, which were obtained numerically as the eigenvalues of the
$N_\mathrm{DOF}\times N_\mathrm{DOF}$ matrices that result as each
operator discretized on the basis of the method introduced in Section
\ref{sec:4} with $6$ patches, $N=16$, $N^\beta=100$ and $p=8$, are
displayed in Figure~\ref{eig1}. (Matrices for the various
operators were obtained by applying the discretized operators
described in Section \ref{sec:4} to the canonical basis of
$\mathbb{C}^{N_\mathrm{DOF}}$, and, for simplicity, the eigenvalues of
the resulting matrices were obtained by means of Matlab's function
{\it eig}.)

\begin{figure}[htb]
\centering
\includegraphics[scale=0.35]{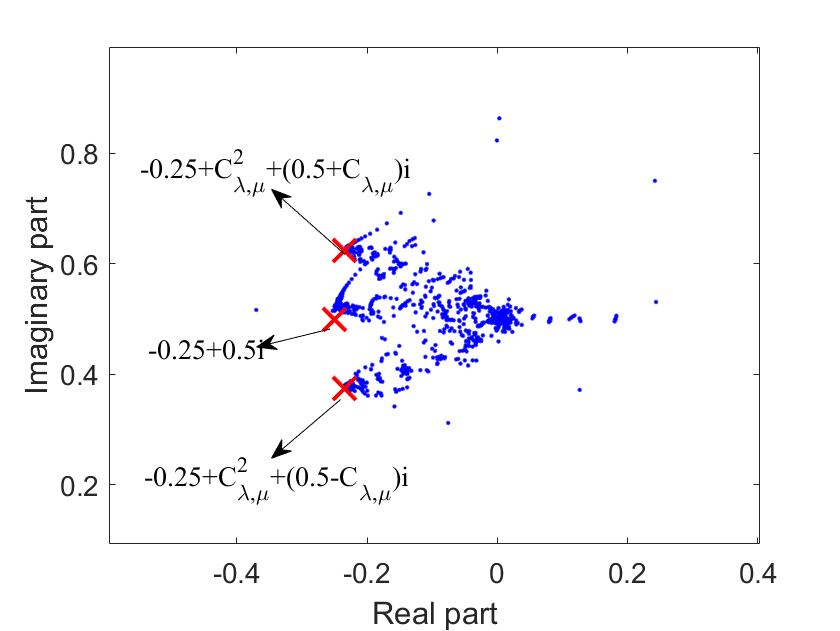}
\caption{Eigenvalue distribution for the integral operator $i\left(\frac{I}{2}-K'\right)+N\mathcal{R}$ in (\ref{BIER1}).}
\label{eig2}
\end{figure}

\subsection{Regularized boundary integral equation I: closed-surface case}
\label{sec:3.2}

Relying on the studies presented in Section \ref{sec:3.1} of the
spectra of various relevant elastic-scattering integral operators,
this section proposes regularized combined field equations that make
use of the single-layer operator $\mathcal{R}:=S_{i\omega_1}$
($\omega_1>0$) for the ``imaginary-frequency'' $i\omega_1$---in
addition to the aforementioned double-layer and hypersingular
operators $K'$ and $N$. For simplicity, we assume
$\omega_1=\omega$. Thus, replacing the scattered field
representation~\eqref{sol1} by the expression \be
\label{sol1r}
u(x)=(\mathcal{D}\mathcal{R}-i\eta\mathcal{S})(\psi)(x), \quad
x\in D, \en  we obtain  the regularized
integral equation \be
\label{BIER1}
\left[i\eta
  \left(\frac{I}{2}-K'\right)+N\mathcal{R}\right](\psi)= F
\quad\mbox{on}\quad \Gamma, \en instead of the classical combined
field equation (\ref{BIE1}). The favorable properties of
equation~\eqref{BIER1} are described in the following theorem.

\begin{theorem}
\label{theorem2}
The regularized integral equation (\ref{BIER1}) is uniquely solvable.
The spectrum of the regularized combined field integral operator on
the left hand side of that equation consists of three non-empty
sequences of eigenvalues which converge to
$-1/4+C_{\lambda,\mu}^2+i\eta(1/2+C_{\lambda,\mu})$, $-1/4+i\eta/2$
and $-1/4+C_{\lambda,\mu}^2+i\eta(1/2-C_{\lambda,\mu})$, respectively.
\end{theorem}
\begin{proof}
  We need to show that the homogeneous equation of (\ref{BIER1}) only
  admits the trivial solution. Let us call $u^+$ (resp. $u^-$) the
  potential defined for $x\in D$ (resp. $x\in\Omega$) by the right
  hand side of Equation (\ref{sol1r}). Clearly, $u^+$ is a radiative
  solution to the elastic problem in $D$ with $T(\pa,\nu)u^+=0$ on
  $\Gamma$. We conclude that $u^+=0$ everywhere outside $\Omega$. From
  the classical jump relations for the boundary values of layer
  potentials, we see that \ben u^-=-S_{i\omega}(\psi),\quad
  T(\pa,\nu)u^-=i\eta\psi.  \enn Applying Betti's formula~\cite{BHSY},
  we then obtain  \ben i\eta\int_\Gamma
  S_{i\omega}(\psi)\ov{\psi}ds=\int_\Omega
  \left(\frac{\mu}{2}\left|\nabla u^-+\nabla^\top u^-\right|^2+
    \lambda|\nabla\cdot u^-|^2- \rho \omega^2|u^-|^2 \right)\,dx, \enn
  and therefore, \ben \int_\Gamma S_{i\omega}(\psi)\ov{\psi}ds=0.
  \enn It is known that $S_{i\omega}$ is positive definite~\cite[Lemma
  6.2]{AK02}, that is, there exists some positive constant $c>0$ such
  that \ben \int_\Gamma S_{i\omega}(\psi)\ov{\psi}ds\ge
  c\|\psi\|_{H^{-1/2}(\Gamma)^3}^2.  \enn This implies that $\psi=0$
  on $\Gamma$.

  Noting that \ben i\eta
  \left(\frac{I}{2}-K'\right)+N\mathcal{R}=i\eta
  \left(\frac{I}{2}-K'\right)+N S+N(S_{i\omega}-S), \enn and since
  $(S_{i\omega}-S):H^{-1/2}(\Gamma)^3\rightarrow H^{1/2}(\Gamma)^3$ is
  a compact operator (in view of its kernel's smoothness), the claims
  concerning accumulation points of eigenvalue sequences of the
  combined integral operator~(\ref{BIER1}) follows from the results
  presented in Section~\ref{sec:3.1}. The proof is now complete.
\end{proof}

To illustrate Theorem~\ref{theorem2} we utilize once again the
unit-ball scattering problem considered in Section~\ref{sec:3.1}. The
spectrum of the corresponding regularized combined field operator is
displayed in Figure~\ref{eig2}. Clearly the eigenvalues accumulate as
prescribed by the theorem, and, in particular, they do not accumulate
either at zero or infinity.

\subsection{Regularized boundary integral equation II: open-surface case}
\label{sec:3.3}

In our treatment of an open surface $\Gamma$ we assume, for
simplicity, that the surface $\Gamma$, its edge, and the right hand
side in equation~(\ref{BIE2}) are infinitely smooth. Under such
assumptions, the singular character of the solution $\varphi$ is given
by~\cite{CDD03} \ben \varphi= \psi\,d^{1/2}, \enn where $\psi$ is an
infinitely differentiable function in a neighborhood of the edge, up
to and including the edge, and where $d$ denotes the distance to the
edge. In view of this result we introduce a weight function $w(x)$
which is smooth, positive and non-vanishing across the interior of the
surface, and which, up to a factor that is $C^\infty$ throughout
$\Gamma$ (including the edge) has square-root asymptotic edge behavior
\ben w\sim d^{1/2}\quad\mbox{around the edge of }\Gamma.  \enn Then we
define the weighted operator \be
\label{weightN}
N_w(\psi)=N(w\psi), \en so that for functions $F$ that are smooth on
$\Gamma$, up to and including the edge, the solution of the equation
\be
\label{weightBIE2}
N_w(\psi)=F \quad\mbox{on}\quad\Gamma,
\en
is also smooth throughout the surface. In view of the spectral properties of the closed-surface composite operator $N S$, we consider the composite operator $N_w S_w$ and the corresponding
equation
\be
\label{BIER2}
N_w S_w(\psi)=F \quad\mbox{on}\quad\Gamma.  \en Here $S_w$ is a
weighed version of the operator $S$, \ben S_w(\psi)=S(\psi/w), \enn
(which can also be used for treatment of the scattering problem under
Dirichlet boundary conditions~\cite{BL12,BL13,BXY191,LB15}; see also
Remark~\ref{RDirichlet} and Figures~\ref{FigExp3.12}
and~\ref{FigExp3.3} in Section~\ref{sec:5}).

As shown in~\cite{BL12,LB15}, the equation analogous to~(\ref{BIER2})
for the 2D acoustic open-arc case is a second-kind equation. Further,
the numerical results presented in \cite{BL12,BL13} for 2D/3D acoustic
problems and in \cite{BXY191} for 2D elastic problems show that, in
the cases considered in those contributions, equation (\ref{BIER2})
requires significantly smaller numbers of GMRES iterations than
equation (\ref{weightBIE2}) for convergence to a given residual
tolerance. For our numerical study of the spectrum of the operator
$N_wS_w$ in the present elastic case we consider operators associated
with the problem of elastic scattering by a unit disc (using the same
parameters in Section \ref{sec:3.1}). Figure~\ref{eig3} displays
numerical values of the eigenvalues of the operator $N_wS_w$, which
were obtained by applying the discretization method introduced in
Section \ref{sec:4} with 5 patches. This figure clearly suggests that
the eigenvalues of $N_wS_w$ are at least bounded away from infinity,
although they also appear to approach the origin. As demonstrated in
Figure~\ref{FigExp3.11}, reduction in iteration numbers are observed
for three-dimensional open-surface elastic problems that are analogous
to those obtained for the corresponding closed-surface elastic case
(Figures~\ref{FigExp2.1} through~\ref{FigExp2.3}).

\begin{figure}[htb]
\centering
\includegraphics[scale=0.35]{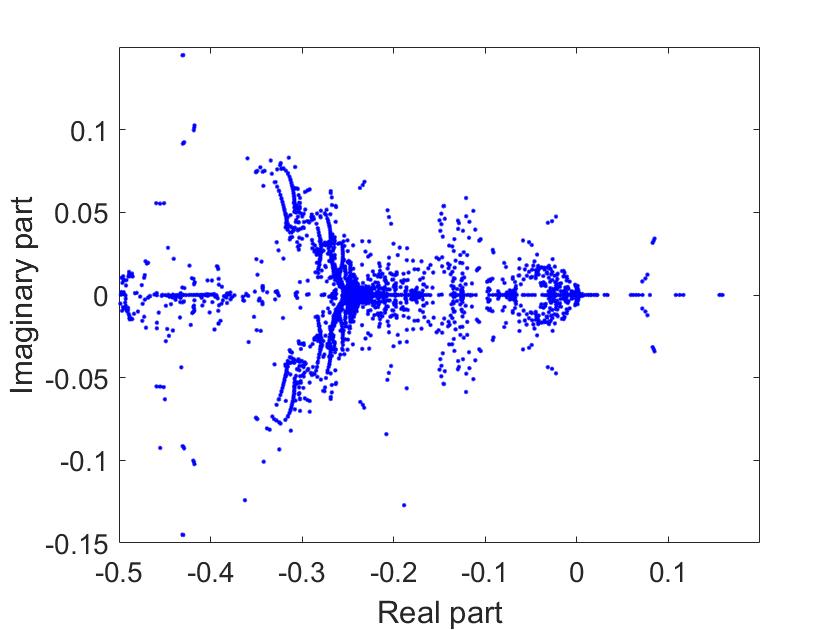}
\caption{Eigenvalue distribution for the integral operator $N_wS_w$ in
  (\ref{BIER2}).}
\label{eig3}
\end{figure}

\begin{remark}
  \label{RDirichlet}
  The regularization techniques introduced for the Neumann problem can
  also be applied for the problems of scattering under Dirichlet
  boundary conditions \ben u=G \quad\mathrm{on}\quad \Gamma.  \enn For
  the Dirichlet problem the solution can be expressed as a
  single-layer potential \ben u(x)=\mathcal{S}(\varphi)(x),\quad x\in
  D, \enn which results in the boundary integral equation \ben
  S(\varphi)=G \quad\mathrm{on}\quad \Gamma.  \enn The singular
  character of the solution $\varphi$, which is given by
  $\varphi\sim \psi\, d^{-1/2} $~\cite{CDD03} where $\psi$ is an
  infinitely differentiable function throughout $\Gamma$, up to and
  including the edge, leads us to consider the weighted integral
  equation \be
\label{BIED}
S_w(\psi)=G \quad\mathrm{on}\quad \Gamma.  \en As in the Neumann case,
further, we can also consider the combined operator $N_wS_w$ and the
corresponding equation \be
\label{BIEDR}
N_wS_w(\psi)=N_w(G) \quad\mathrm{on}\quad \Gamma, \en for the
Dirichlet problem---although, as demonstrated in
Figure~\ref{FigExp3.12}, the single layer formulation~\eqref{BIED}
already requires small iteration numbers, and no improvements in
iteration numbers result in this case from use of the $N_wS_w$
formulation.
\end{remark}

\subsection{Strong-singularity and hyper-singularity regularization}
\label{sec:3.4}

As noted in Sections~\ref{sec:1},~\ref{sec:2.2} and~\ref{sec:2.3}, the
integral operators $K'$, $N$ and $N_w$ are strongly singular and
hyper-singular, respectively. This section expresses the strongly
singular and hyper-singular boundary integral operators (\ref{BIER1})
and (\ref{BIER2}) in terms of compositions of operators of
differentiation in directions tangential to $\Gamma$ and
weakly-singular integral operators~\cite{BXY19,YHX}. Using this
reformulation together with efficient numerical implementations of
weakly-singular and tangential differentiation operators and the
linear algebra solver GMRES then leads to the proposed elastic-wave
solvers.

The traction operator can be expressed in the form \be \label{Tform2}
T(\pa,\nu)u(x)= (\lambda+\mu)\nu(\nabla \cdot u) + \mu\pa_{\nu}u + \mu
M(\pa,\nu)u \en where the operator $M(\pa,\nu)$, whose elements are
also called G\"unter derivatives, is defined by \ben M(\pa,\nu)u(x)=
\pa_{\nu}u -\nu(\nabla\cdot u)+\nu\times \,\curl\,u.  \enn Letting
$M(\pa_x,\nu_x)=[m_x^{ij}]_{i,j=1}^3$, it is easy to check that \ben
m_x^{ij}=\pa_{x_i}\nu_x^j-\pa_{x_j}\nu_x^i=-m_x^{ji}
\quad\mbox{for}\quad i,j=1,2,3, \enn and \ben
\label{M3}
M(\pa,\nu) =\begin{pmatrix}
0 & -\widetilde{\pa}_3 & \widetilde{\pa}_2 \\
\widetilde{\pa}_3 & 0 & -\widetilde{\pa}_1 \\
-\widetilde{\pa}_2 & \widetilde{\pa}_1 & 0
\end{pmatrix}, \enn where $\widetilde{\pa}_i,i=1,2,3$ are the
components of $\nu\times\nabla$, i.e.,
$\nu\times\nabla=(\widetilde{\pa}_1, \widetilde{\pa}_2,
\widetilde{\pa}_3)^\top$. Let $\nabla^S$ denote the surface gradient:
\ben \nabla^Su=\nabla u-\nu\pa_\nu u.  \enn Then we have
$\nu\times\nabla=\nu\times\nabla^S$. Writing
$\nu\times\nabla^S=(\widetilde{\pa}_1^S, \widetilde{\pa}_2^S,
\widetilde{\pa}_3^S)^\top$, we obtain \ben M(\pa,\nu) =\begin{pmatrix}
  0 & -\widetilde{\pa}_3^S & \widetilde{\pa}_2^S \\
  \widetilde{\pa}_3^S & 0 & -\widetilde{\pa}_1^S \\
  -\widetilde{\pa}_2^S & \widetilde{\pa}_1^S & 0
\end{pmatrix}.
\enn

The following lemma can be established as in~\cite{BXY19}, and we omit
the proof here.
\begin{lemma}
\label{RegClose}
The boundary integral operator $K'$ can be expressed in the form
\be
\label{operatorK}
K' &=& K_1+M(\pa,\nu)K_2,
\en
where
\ben
K_1(\varphi)(x) &=& \int_\Gamma \left\{\pa_{\nu_x}\gamma _{k_s}(x,y)I- \nu_x \nabla _x^\top [\gamma_{k_s}(x,y)-\gamma_{k_p}(x,y)] \right\}\varphi(y)ds_y\quad\mbox{and}\\
K_2(\varphi)(x) &=& \int_\Gamma \left[2\mu E(x,y) - \gamma
_{k_s}(x,y)I\right]\varphi(y)ds_y.
\enn
For the hyper-singular operator $N$, in turn, we have
\be
\label{operatorN}
N=N_1+M(\pa,\nu)N_2M(\pa,\nu)+\mathcal{T}_2N_3\mathcal{T}_1+M(\pa,\nu)N_4+N_5M(\pa,\nu),
\en
where
\ben
N_1(\varphi)(x)&=& -\rho\omega^2\int_\Gamma\left[ \gamma_{k_s}(x,y)(\nu_x\nu_y^\top-\nu_x^\top\nu_yI)- \gamma_{k_p}(x,y)\nu_x\nu_y^\top\right]\varphi(y)ds_y,\nonumber\\
N_2(\varphi)(x) &=& \int_\Gamma\left[4\mu^2 E(x,y) -3\mu\gamma_{k_s}(x,y)I\right]\varphi(y)ds_y,\nonumber\\
N_3(\varphi)(x) &=& \mu\int_\Gamma \gamma_{k_s}(x,y)\varphi(y)ds_y,\nonumber\\
N_4(\varphi)(x) &=& \int_\Gamma \left\{\mu\pa_{\nu_y}\gamma_{k_s}(x,y)I- 2\mu\nabla_y[\gamma_{k_s}(x,y)-\gamma_{k_p}(x,y)]\nu_y^\top \right\}\varphi(y)ds_y,\nonumber\\
N_5(\varphi)(x) &=& \int_\Gamma \left\{\mu\pa_{\nu_x}\gamma_{k_s}(x,y)I- 2\mu\nu_x\nabla_x^\top [\gamma_{k_s}(x,y)-\gamma_{k_p}(x,y)] \right\} \varphi(y)ds_y,
\enn
and where, for a scalar field $v$ and a vector field $V$, the operators $\mathcal{T}_1$ and $\mathcal{T}_2$ in (\ref{operatorN}) are defined by
\ben
\mathcal{T}_1v=\nu\times\nabla^Sv, \quad\mathcal{T}_2V=(\nu\times\nabla^S)\cdot V.
\enn
The kernels of the integral operators $K_i, i=1,2$ in
(\ref{operatorK}) and $N_j, j=1,\cdots,5$ in (\ref{operatorN}) are all
at-most weakly-singular.
\end{lemma}

Noting that $w^2(x)$ is a smooth function of $x$ throughout $\Gamma$
which vanishes at the edge of $\Gamma$, the following open-surface
version of the previous lemma can similarly be established.
\begin{lemma}
\label{RegOpen}
The hyper-singular operator $N_w$ can be expressed in the form \be
\label{operatorNw}
N_w=N_1^w+M(\pa,\nu)N_2^w\mathcal{T}^w+\mathcal{T}_2N_3^w\mathcal{T}_1^w
+M(\pa,\nu)N_4^w+N_5^w\mathcal{T}^w, \en where \ben
N_i^w(\varphi)&=& N_i(w\varphi),\quad i=1,4,\\
N_j^w(\varphi)&=& N_j(\varphi/w),\quad j=2,3,5, \enn and where, for a
scalar field $v$ and a vector field $V$, the operators
$\mathcal{T}_1^w$ and $\mathcal{T}^w$ are given by \ben
\mathcal{T}_1^wv=w^2\nu\times\nabla^Sv+
\frac{v}{2}\nu\times\nabla^S(w^2), \enn \ben
\mathcal{T}^wV=w^2M(\pa,\nu)V+M(\pa,\nu)(w^2)\frac{V}{2}, \enn
respectively. The kernels of the integral operators
$N_j^w, j=1,\cdots,5$ in (\ref{operatorNw}) are all at-most
weakly-singular.
\end{lemma}

\section{Numerical implementation}
\label{sec:4}

In view of the integral-operator formulations presented in
Section \ref{sec:3.4}, a numerical version of the regularized integral
operators introduced in Sections~\ref{sec:2.2} and~\ref{sec:3.3} can
be obtained as a sum of (possibly multiple) compositions of numerical
operators of two types, namely, (i)~Integral operators of the forms
\be
\label{weakly} \mathcal{H}\varphi(x)&=&\int_\Gamma
  H(x,y)\varphi(y)ds_y,\\
\label{weakly1}
\widehat{\mathcal{H}}^1(x)&=&\int_\Gamma H(x,y)\varphi(y)w(y)ds_y,\\
\label{weakly2}
\widehat{\mathcal{H}}^2(x)&=&\int_\Gamma H(x,y)\varphi(y)/w(y)ds_y,
\en in which the kernel $H(x,y)$ is weakly singular, and
(ii)~Differentiation operators for the evaluation of the surface
gradient of a given smooth function defined on $\Gamma$. Here, the
integral (\ref{weakly}) is related to closed-surface problems and the
integrals (\ref{weakly1}) and (\ref{weakly2}) are related to
open-surface problems. This section presents algorithms for numerical
evaluation of operators of these types, including a
rectangular-polar~\cite{BG18} Chebyshev-based quadrature method for
weakly singular operators $\mathcal{H}$, $\widehat{\mathcal{H}}^1$ and
$\widehat{\mathcal{H}}^2$ as well as Chebyshev-based differentiation
algorithms. In all, the regularized iterative open- and closed-surface
solvers rely on
\begin{itemize}
\item[(1).] A partition of the scattering surface $\Gamma$ into a set of
  non-overlapping logically-quadrilateral parametrized patches;
\item[(2).] High-order integration rules based on Chebyshev polynomials,
  Fejer's first quadrature rule, and ``rectangular-polar'' changes of
  variables which produce accurate approximations of the integral
  operators with weakly-singular kernels;
\item[(3).] Chebyshev-based differentiation rules; and,
\item[(4).] The iterative linear algebra solver GMRES, for solution of
  the discrete versions of Eqs. (\ref{BIER1}) and (\ref{BIER2}).
\end{itemize}

The methods for evaluation of the weakly singular and differentiation
operators in closed-surface cases differ somewhat from their
open-surface counterparts. Accordingly, Sections~\ref{sec:4.1} and
\ref{sec:4.2} present algorithms for the tasks~(1) through~(3) above
in the closed- and open-surface cases,
respectively. Section~\ref{sec:4.3}, finally, presents overall
pseudo-codes for the complete scattering algorithms.

\subsection{Closed-surface case}
\label{sec:4.1}

\subsubsection{Surface partitioning and discretization}
\label{sec:4.1.1}
The proposed numerical method evaluates the necessary weakly-singular
operators $\mathcal{H}$ on the basis of the Chebyshev-based
rectangular-polar solver developed in \cite{BG18}. We thus assume the
scattering surface has been partitioned into a set of $M$
non-overlapping ``logically-quadrilateral'' parametrized patches
(i.e. patches that can be parametrized from the parameter square
$[-1,1]\times [-1,1]$), which can easily be obtained, for example,
from typical CAD (Computer Aided Design) models. Let, then, the
non-overlapping partition of the scattering surface be given by the
union of logically-rectangular patches $\Gamma_q$, \ben
\Gamma=\bigcup_{q=1}^M \Gamma_q,\quad
\Gamma_q:=\left\{x={\bf r}^q(u,v): [-1,1]^2\rightarrow \R^3\right\}.  \enn Then the integral $\mathcal{H}$ over $\Gamma$ can be
decomposed as a sum of integrals over each one of the patches: \ben
\mathcal{H}(x)=\sum_{q=1}^M \mathcal{H}_q(x),\quad
\mathcal{H}_q(x):=\int_{\Gamma_q} H(x,y)\varphi(y)ds_y, \quad
x\in\Gamma.  \enn

Once the patch structure has been established, a number of
``singular'', ``near-singular'' and ``regular'' integration problems
arise as described in Section~\ref{sec:4.2}, for which specialized
rules are used for accuracy and efficiency. In all cases the numerical
method we use incorporates Fej\'er's first quadrature rule, which
effectively exploits the discrete orthogonality property satisfied by
the Chebyshev polynomials in the Chebyshev meshes. Denoting by
$\tau_j\in [-1,1],j=0,\cdots,N-1$ the $N$ Chebyshev points \ben
\tau_j=\cos\left(\frac{2j+1}{2N}\pi\right),\quad j=0,\cdots,N-1, \enn Using the Cartesian-product discretization $\{u_i=\tau_i| i=1,\cdots,N\}\times \{v_j=\tau_j| j=1,\cdots,N\}$, we choose the discretization points in each patch $\Gamma_q$ according to
\ben x_{ij}^q={\bf r}^q(u_i,v_j),\quad i,j=0,\cdots,N-1.
\enn Then, a given density $\varphi$ with values
$\varphi_{ij}^q = \varphi(x_{ij}^q)$ is approximated by means of the
Chebyshev expansion \ben \varphi(x)\approx \sum_{i,j=0}^{N-1}
\varphi_{ij}^q a_{ij}(u,v),\quad x\in\Gamma_q, \enn where \ben
a_{ij}(u,v)= \frac{1}{N^2}\sum_{m,n=0}^{N-1}\alpha_n\alpha_m
T_n(u_i)T_m(v_j)T_n(u)T_m(v),\quad \alpha_n=\begin{cases} 1, & n=0,\cr
  2, & n\ne 0.
\end{cases}
\enn As is well known the functions $a_{ij}(u,v)$ satisfy the
relations \ben a_{ij}(u_n,v_m)=\begin{cases} 1, & (n,m)=(i,j),\cr 0, &
  \mathrm{otherwise}.
\end{cases}
\enn

\subsubsection{Non-adjacent and adjacent integration}
\label{sec:4.1.2}

The method we use for evaluation of an integral $\mathcal{H}_q$ of the
form~\eqref{weakly} at the discretization points
$x^{\widetilde{q}}_{ij}$ ($\widetilde{q}=1,\cdots,M$) proceeds by
consideration of the distance  \ben
\mbox{dist}_{x,\Gamma_q}:=\min_{(u,v)\in[-1,1]^2}
\left\{|x-{\bf r}^q(u,v)| \right\} \enn between the point
$x$ and the patch $\Gamma_q$. Denote the index sets
\be
\label{indexseta}
\mathrm{I_a}&:=& \{(q,\widetilde{q},i,j)| \mbox{dist}_{x^{\widetilde{q}}_{ij},\Gamma_q}\le \tau,\;q,\widetilde{q}=1,\cdots,M, \; i,j=1,\cdots,N\}, \\
\label{indexsetna}
\mathrm{I_{na}}&:=& \{(q,\widetilde{q},i,j)| \mbox{dist}_{x^{\widetilde{q}}_{ij},\Gamma_q}> \tau,\;q,\widetilde{q}=1,\cdots,M, \; i,j=1,\cdots,N\},
\en
where $\tau$ is some tolerance (in this paper, we use $\tau=0.1$).

In the "non-adjacent" integration case,
in which the point $x^{\widetilde{q}}_{ij}$ is far from the
integration patch (i.e., $(q,\widetilde{q},i,j)\in\mathrm{I_{na}}$), the integrand
$\mathcal{H}_q(x^{\widetilde{q}}_{ij})$ is smooth. Then this integral
can be accurately evaluated by means of
Fej\'er's first quadrature rule  \be
\label{nonadjencent}
\mathcal{H}_q(x^{\widetilde{q}}_{ij}) &=& \int_{\Gamma_q} H(x^{\widetilde{q}}_{ij},y)\varphi(y)ds_y \nonumber\\
&=& \int_{-1}^1\int_{-1}^1 H(x^{\widetilde{q}}_{ij},{\bf r}^q(u,v)) \varphi({\bf r}^q(u,v)) J^q(u,v) \,dudv \nonumber\\
&\approx& \sum_{m,n=0}^{N-1}
H(x^{\widetilde{q}}_{ij},{\bf r}^q(u_n,v_m)) \varphi_{nm}^q
J^q(u_n,v_m)w_nw_m, \en where $J^q(u,v)$ denotes the surface Jacobian
and $w_j,j=0,\cdots,N-1$ are the quadrature weights \ben
w_j=\frac{2}{N}\left(1-2\sum_{l=1}^{\lfloor N/2\rfloor}
  \frac{1}{4l^2-1}\cos(lu_j)\right),\quad j=0,\cdots,N-1.  \enn

In the "adjacent" integration case, in which the point
$x^{\widetilde{q}}_{ij}$ either lies within the integration patch or
is "close" to it (i.e., $(q,\widetilde{q},i,j)\in\mathrm{I_{a}}$), in turn, the
problem of evaluation of $\mathcal{H}_q(x^{\widetilde{q}}_{ij})$
presents a challenge in view of the singularity or nearly-singularity
of its kernel. To tackle this difficulty we apply a change of
variables whose derivatives vanish at the singularity or, for nearly
singular problems, at the point in the integration patch that is
closest to the singularity---in either case, the coordinates
$(\widetilde{u}^q,\widetilde{v}^q)\in [-1,1]$ of the point around
which refinements are performed are given by \ben
(\widetilde{u}^q,\widetilde{v}^q)={\arg\min}_{(u,v)\in[-1,1]^2}
\left\{|x^{\widetilde{q}}_{ij}-{\bf r}^q(u,v)| \right\}.
\enn

The quantities $\widetilde{u}^q,\widetilde{v}^q$ can be found by
means of an appropriate minimization algorithm such as the golden
section search algorithm. A  ``rectangular-polar''
change of variables can be constructed on the basis of the
one-dimensional change of variables \ben \xi_\alpha(t)=\begin{cases}
  \alpha+\frac{\mbox{sgn}(t)-\alpha}{\pi}w_p(\pi|t|), & \alpha\ne
  \pm1, \cr
  \alpha-\frac{1+\alpha}{\pi}w_p\left(\pi\frac{|t-1|}{2}\right), &
  \alpha=1, \cr
  \alpha+\frac{1-\alpha}{\pi}w_p\left(\pi\frac{|t+1|}{2}\right), &
  \alpha=-1.
\end{cases}
\enn Here $w_p(t)$ is a function depending on a constant $p\ge 2$
given by \ben
w_p(t)=2\pi\frac{[\eta_p(t)]^p}{[\eta_p(t)]^p+[\eta_p(2\pi-t)]^p},\quad
0\le t\le 2\pi, \enn where \ben
\eta_p(t)=\left(\frac{1}{p}-\frac{1}{2}\right)
\left(\frac{\pi-t}{\pi}\right)^3
+\frac{1}{p}\left(\frac{\pi-t}{\pi}\right)+ \frac{1}{2}.  \enn It is
easy to check that the derivatives of $w_p(t)$ up to order $p-1$
vanish at the endpoints. Applying the Chebyshev expansion of the
density $\varphi$, the above change of variables and the Fej\'er's
first quadrature rule, we obtain
 \be
\label{adjencent}
\mathcal{H}_q(x^{\widetilde{q}}_{ij}) &=& \int_{\Gamma_q} H(x^{\widetilde{q}}_{ij},y)\varphi(y)ds_y \nonumber\\
&\approx& \sum_{n,m=0}^{N-1}\varphi_{nm}^q \int_{-1}^1\int_{-1}^1 H(x^{\widetilde{q}}_{ij},{\bf r}^q(u,v)) J^q(u,v)a_{nm}(u,v) \,dudv \nonumber\\
&=& \sum_{n,m=0}^{N-1}\varphi_{nm}^q \int_{-1}^1\int_{-1}^1 \widetilde{H}(x^{\widetilde{q}}_{ij},s,t) \widetilde{J}^q(s,t) \widetilde{a}_{nm}(s,t)\xi_{\widetilde{u}^q}'(s) \xi_{\widetilde{v}^q}'(t) \,dsdt \nonumber\\
&\approx& \sum_{n,m=0}^{N-1}A_{ij,nm}^{\widetilde{q},q}\varphi_{nm}^q
\en
where
\be
\label{matrixA}
A_{ij,nm}^{\widetilde{q},q} &=&\sum_{l_1,l_2}^{N^\beta-1} \widetilde{H}(x^{\widetilde{q}}_{ij},\widetilde{u}_{l_1},\widetilde{u}_{l_2}) \widetilde{J}^q(\widetilde{u}_{l_1},\widetilde{u}_{l_2}) \widetilde{a}_{nm}(\widetilde{u}_{l_1},\widetilde{u}_{l_2}) \xi_{\widetilde{u}^q}'(\widetilde{u}_{l_1}) \xi_{\widetilde{v}^q}'(\widetilde{u}_{l_2}) \widetilde{w}_{l_1}\widetilde{w}_{l_2},
\en
with
\ben
\widetilde{H}(x^{\widetilde{q}}_{ij},s,t)&=& H(x^{\widetilde{q}}_{ij}, {\bf r}^q(\xi_{\widetilde{u}^q}(s),\xi_{\widetilde{v}^q}(t))), \\ \widetilde{J}^q(s,t)&=& J^q(\xi_{\widetilde{u}^q}(s),\xi_{\widetilde{v}^q}(t)), \\ \widetilde{a}_{nm}(s,t)&=& a_{nm}(\xi_{\widetilde{u}^q}(s),\xi_{\widetilde{v}^q}(t)),
\enn
and where  the quadrature nodes and weights are given by
\ben
\widetilde{u}_j=\cos\left(\frac{2j+1}{2N^\beta}\pi\right),\quad j=0,\cdots,N^\beta-1,
\enn
and
\ben
\widetilde{w}_j=\frac{2}{N^\beta}\left(1-2\sum_{l=1}^{\lfloor N^\beta/2\rfloor} \frac{1}{4l^2-1}\cos(l\widetilde{u}_j)\right),\quad j=0,\cdots,N^\beta-1.
\enn
Using sufficiently large numbers $N^\beta$ of discretization points along the $u$ and $v$ directions to accurately resolve the challenging integrands, all singular and nearly singular problems can be treated with high accuracy under discretizations that are not excessively fine.

\subsubsection{Evaluation of surface gradients}
\label{sec:4.1.3}

Now we describe the implementation we use for the evaluation of
the surface gradient $\nabla^S$, from which the needed
surface-differentiation operators
$M(\pa,\nu),\mathcal{T}_1,\mathcal{T}_2$
can be extracted. On each patch $\Gamma_q$, the surface gradient of a
given density $\psi(u,v)=\varphi({\bf r}^q(u,v))$ is given by \ben \nabla_x^S\psi &=&
g^{11}\frac{d\psi}{d u}\frac{d {\bf r}^q}{d u}
+g^{12}\frac{d\psi}{d u}\frac{d {\bf r}^q}{d v}
+g^{21}\frac{d\psi}{d v}\frac{d {\bf r}^q}{d u}
+g^{22}\frac{d\psi}{d v}\frac{d {\bf r}^q}{d v},
\enn where $g^{ij}, i,j=1,2$ denote the components of the inverse of the first
fundamental matrix $G=[g_{ij}]_{i,j=1}^2$ with \ben
&& g_{11}=\frac{d {\bf r}^q}{d u}\cdot \frac{d {\bf r}^q}{d u},\quad g_{12}=\frac{d {\bf r}^q}{d u}\cdot \frac{d {\bf r}^q}{d v}, \\
&& g_{21}=\frac{d {\bf r}^q}{d v}\cdot \frac{d
  {\bf r}^q}{d u},\quad g_{22}=\frac{d
  {\bf r}^q}{d v}\cdot \frac{d
  {\bf r}^q}{d v}.  \enn The quantities
$\frac{d\psi}{d u},\frac{d\psi}{d v}$ can be easily
evaluated by means of term-by-term differentiation of the Chebyshev
expansion of $\psi$. Therefore, we have
\ben
(\nabla_x^S\psi)\Big|_{x=x^{q}_{ij}} =\sum_{n,m=0}^{N-1} B_{ij,nm}^q\varphi_{nm}^q,
\enn
where
\be
\label{matrixB}
B_{ij,nm}^q=\left(g^{11}\frac{da_{nm}}{d u}\frac{d {\bf r}^q}{d u}
+g^{12}\frac{da_{nm}}{d u}\frac{d {\bf r}^q}{d v}
+g^{21}\frac{da_{nm}}{d v}\frac{d {\bf r}^q}{d u}
+g^{22}\frac{da_{nm}}{d v}\frac{d {\bf r}^q}{d v}\right)\Big|_{u=u_i,v=v_j}.
\en

\subsection{Open-surface case}
\label{sec:4.2}
\subsubsection{Surface partitioning, discretization and integration}
\label{sec:4.2.1}

As we did for closed surfaces, here we assume the open scattering
surface has been partitioned into a set of $M$ non-overlapping
logically-rectangular patches $\Gamma_q$, \ben \Gamma=\bigcup_{q=1}^M
\Gamma_q,\quad \Gamma_q:=\left\{x=\widehat{\bf r}^q(u,v):
  [-1,1]^2\rightarrow \R^3\right\}.  \enn Then the integrals
$\widehat{\mathcal{H}}^1$ and $\widehat{\mathcal{H}}^2$ can be
decomposed as sums of integrals over each one of the patches: \be
\label{weaklyw1}
\widehat{\mathcal{H}}^1(x)=\sum_{q=1}^M \widehat{\mathcal{H}}^1_q(x),\quad \widehat{\mathcal{H}}^1_q(x):=\int_{\Gamma_q} H(x,y)\varphi(y)w(y)ds_y, \quad x\in\Gamma,\\
\label{weaklyw2}
\widehat{\mathcal{H}}^2(x)=\sum_{q=1}^M
\widehat{\mathcal{H}}^2_q(x),\quad
\widehat{\mathcal{H}}^2_q(x):=\int_{\Gamma_q}
H(x,y)\varphi(y)/w(y)ds_y, \quad x\in\Gamma.  \en

In view of the weight function $w\sim d^{1/2}$ that is present in the
integrands of both~\eqref{weaklyw1} and~\eqref{weaklyw2}, a direct
application in the present context of the integration method proposed
in Section~\ref{sec:4.1.2} only yields accuracy of low order. To
demonstrate this fact we consider the integrals \ben I_1=\int_{-1}^1
\cos(t)\sqrt{1-t^2}\,dt,\quad I_2=\int_{-1}^1
\frac{\cos(t)}{\sqrt{1-t^2}}\,dt, \enn where the term
$w(t) = \sqrt{1-t^2}$ is the singular weight function in this case. As
demonstrated in Table \ref{compI1I2}, applications of the Fejer's
first quadrature rule to the integrals $I_1$ and $I_2$ only yield
third- and and first-order convergence, respectively.

\begin{table}[htb]
  \caption{Errors in the evaluation of the integrals $I_1$ and $I_2$ by means of Fejer's first quadrature rule.}
\centering
\begin{tabular}{|c|c|c|c|c|}
\hline
$N$ & $I_1$ & Order & $I_2$ & Order \\
\hline
5  & 9.06E-4 & --   & 3.27E-2 & -- \\
\hline
10 & 1.65E-4 & 2.46 & 1.55E-2 & 1.08 \\
\hline
15 & 4.87E-5 & 3.01 & 1.04E-2 & 0.98 \\
\hline
20 & 2.12E-5 & 2.89 & 7.75E-3 & 1.02 \\
\hline
30 & 6.30E-6 & 2.99 & 5.17E-3 & 1.00 \\
\hline
\end{tabular}
\label{compI1I2}
\end{table}

To evaluate of the integrals (\ref{weaklyw1}) and (\ref{weaklyw2})
with high accuracy order we introduce the  change of
variables \ben
u=\eta_u^q(s)=
\begin{cases}
s, & \mathrm{No\;edge\;on\;}u,\cr
\cos(\frac{\pi}{2}(1-s)), & \mathrm{Edges\;at\;}u=\pm 1,\cr
1-2\cos(\frac{\pi}{4}(1+s)), & \mathrm{Edge\;at\;}u=-1\mathrm{\;only},\cr
2\cos(\frac{\pi}{4}(1-s))-1, & \mathrm{Edge\;at\;}u=1\mathrm{\;only},
\end{cases}
\enn which maps the interval $[-1,1]$ to itself. Incorporating this
change of variables we obtain \be
\label{weaklyw11}
 \widehat{\mathcal{H}}^1_q(x)= \int_{-1}^1\int_{-1}^1 H(x,\widehat{\bf r}^q(\eta_u^q(s),\eta_v^q(t))) \varphi(\widehat{\bf r}^q(\eta_u^q(s),\eta_v^q(t))) \widehat{J}^q(\eta_u^q(s),\eta_v^q(t))\widetilde{w}_1(s,t) \,dsdt,
\en
and
\be
\label{weaklyw21}
\widehat{\mathcal{H}}^2_q(x)= \int_{-1}^1\int_{-1}^1 H(x,\widehat{\bf
  r}^q(\eta_u^q(s),\eta_v^q(t))) \varphi(\widehat{\bf
  r}^q(\eta_u^q(s),\eta_v^q(t)))
\widehat{J}^q(\eta_u^q(s),\eta_v^q(t))\widetilde{w}_2(s,t) \,dsdt, \en
where \ben
w_1(s,t) &=& \frac{d\eta_u^q(s)}{ds}\frac{d\eta_v^q(t)}{dt}w(\widehat{\bf r}^q(\eta_u^q(s),\eta_v^q(t))),\\
w_2(s,t) &=&
\frac{d\eta_u^q(s)}{ds}\frac{d\eta_v^q(t)}{dt}[w(\widehat{\bf
  r}^q(\eta_u^q(s),\eta_v^q(t)))]^{-1}, \enn and
$\widehat{J}^q(\eta_u^q(s),\eta_v^q(t))$ denotes the surface
Jacobian. It is easily checked that the integrands in~(\ref{weaklyw1})
and~(\ref{weaklyw2}) equal the products of the weakly singular kernel
$H(x,\widehat{\bf r}^q(\eta_u^q(s),\eta_v^q(t)))$ multiplied by a
smooth function.

Using the Cartesian-product discretization
$\{s_i=\tau_i| i=1,\cdots,N\}\times \{t_j=\tau_j| j=1,\cdots,N\}$, we
choose the discretization points in each patch $\Gamma_q$ according to
\ben \widehat{x}_{ij}^q=\widehat{\bf
  r}^q(\eta_u^q(s_i),\eta_v^q(t_j)),\quad i,j=0,\cdots,N-1.  \enn
Then, a given density $\varphi$ with values
$\widehat{\varphi}_{ij}^q = \varphi(\widehat{x}_{ij}^q)$ is
approximated by means of the Chebyshev expansion \ben
\varphi(x)\approx \sum_{i,j=0}^{N-1} \widehat{\varphi}_{ij}^q
a_{ij}(s,t),\quad x=\widehat{\bf
  r}^q(\eta_u^q(s),\eta_v^q(t))\in\Gamma_q, \enn and the non-adjacent
and adjacent evaluation of the integrals (\ref{weaklyw11}) and
(\ref{weaklyw21}) with respect to $(s,t)$ at the discretization points
$\widehat{x}_{ij}^{\widetilde{q}}, q=1,\cdots,M, i,j=1,\cdots,N$ is
then produced, with high-order accuracy, by means of the numerical
strategy presented in Section \ref{sec:4.1.2}.

\subsubsection{Evaluation of surface gradients}
\label{sec:4.2.2}

Finally, we describe the implementation we use for the evaluation of
the surface gradient $\nabla^S$ and associated operators $M(\pa,\nu)$,
$\mathcal{T}^w$ and $\mathcal{T}_1^w$ for open-surface
problems. Incorporating the open-surface change of variables
introduced in Section~\ref{sec:4.2.1}, the surface gradient of a given
density $\psi(s,t)=\varphi({\bf r}^q(\eta_u^q(s),\eta_v^q(t)))$ on
each patch $\Gamma_q$ is given by \ben \nabla_x^S\psi=
\widehat{g}^{11}\frac{d\psi}{d s}\frac{d {\bf r}^q}{d
  u}\frac{d\eta_u^q(s)}{ds} +\widehat{g}^{12}\frac{d\psi}{d s}\frac{d
  {\bf r}^q}{d v}\frac{d\eta_v^q(t)}{dt}
+\widehat{g}^{21}\frac{d\psi}{d t}\frac{d {\bf r}^q}{d
  u}\frac{d\eta_u^q(s)}{ds} +\widehat{g}^{22}\frac{d\psi}{d t}\frac{d
  {\bf r}^q}{d v}\frac{d\eta_v^q(t)}{dt}, \enn where
$\widehat{g}^{ij}, i,j=1,2$ denote the components of the inverse of
the first fundamental matrix
$\widehat{G}=[\widehat{g}_{ij}]_{i,j=1}^2$ with \ben
&& \widehat{g}_{11}=\frac{d\eta_u^q(s)}{ds}\frac{d\eta_u^q(s)}{ds} \left(\frac{d {\bf r}^q}{d u}\cdot \frac{d {\bf r}^q}{d u}\right),\quad \widehat{g}_{12}=\frac{d\eta_u^q(s)}{ds}\frac{d\eta_v^q(t)}{dt} \left(\frac{d {\bf r}^q}{d u}\cdot \frac{d {\bf r}^q}{d v}\right), \\
&& \widehat{g}_{21}=\frac{d\eta_u^q(s)}{ds}\frac{d\eta_v^q(t)}{dt}
\left(\frac{d {\bf r}^q}{d v}\cdot \frac{d {\bf r}^q}{d
    u}\right),\quad
\widehat{g}_{22}=\frac{d\eta_v^q(t)}{dt}\frac{d\eta_v^q(t)}{dt}
\left(\frac{d {\bf r}^q}{d v}\cdot \frac{d {\bf r}^q}{d v}\right).
\enn Analogously, the quantities $\frac{d\psi}{ds},\frac{d\psi}{dt}$
can be easily evaluated by means of term-by-term differentiation of
the Chebyshev expansion of $\psi$. Therefore, we have \ben
(\nabla_x^S\psi)\Big|_{x=\widehat{x}^{q}_{ij}} =\sum_{n,m=0}^{N-1}
\widehat{B}_{ij,nm}^q\varphi_{nm}^q, \enn where \ben
\widehat{B}_{ij,nm}^q&=&\Big(\widehat{g}^{11}\frac{da_{nm}}{d
  s}\frac{d {\bf r}^q}{d u}\frac{d\eta_u^q(s)}{ds}
+\widehat{g}^{12}\frac{da_{nm}}{d s}\frac{d {\bf r}^q}{d v}\frac{d\eta_v^q(t)}{dt}\\
&\quad&+\widehat{g}^{21}\frac{da_{nm}}{d t}\frac{d {\bf r}^q}{d
  u}\frac{d\eta_u^q(s)}{ds} +\widehat{g}^{22}\frac{da_{nm}}{d
  t}\frac{d {\bf r}^q}{d
  v}\frac{d\eta_v^q(t)}{dt}\Big)\Big|_{s=s_i,t=t_j}.  \enn

\subsection{Overall algorithm pseudocode}
\label{sec:4.3}

Utilizing the concepts presented in Section~\ref{sec:4.1}, the
proposed algorithm for solution of problems of elastic scattering by
closed surfaces is summarized in the following pseudocode. Relying on
Section~\ref{sec:4.2} instead of~\ref{sec:4.1}, the corresponding
pseudocode for open-surface problems is completely analogous, and is
therefore omitted.
\begin{itemize}
\item[I.] Initialization. Input values of $M,N,N^\beta,\tau,p$ and
  construct the surface partitioning and discretization points
  $x_{ij}^q, q=1,\cdots,M, i,j=1,\cdots,N$;
\item[II.] Precomputation.
\begin{itemize}
\item[i.] For all $\widetilde{q},q=1,\cdots,M, i,j=1,\cdots,N$, compute the index sets $\mathrm{I_a}$ and $\mathrm{I_{na}}$, see (\ref{indexseta}) and (\ref{indexsetna});
\item[ii.] Compute the matrices $A_{ij,nm}^{\widetilde{q},q}$, $(q,\widetilde{q},i,j)\in\mathrm{I_{a}}, n,m=1,\cdots,N$ given in (\ref{matrixA}) for adjacent integration;
\item[iii.] Compute the matrices $B_{ij,nm}^{q}$,
  $q=1,\cdots,M, i,j,n,m=1,\cdots,N$ given in (\ref{matrixB}) for
  evaluation of surface gradients;
\end{itemize}
\item[III.] Iterative solution. Use the iterative solver GMRES to
  approximate the solution of the discrete form of the linear equation
  (\ref{BIE1}) or (\ref{BIER1}). The necessary matrix-vector products
  are obtained by suitable compositions and combinations, as detailed
  in Sections~\ref{sec:4.1.2} and~\ref{sec:4.1.3}, of the matrices
  obtained per point~II above.
\end{itemize}

\begin{figure}[htbp]
\centering
\begin{tabular}{ccc}
\includegraphics[height=4.5cm,width=4.5cm]{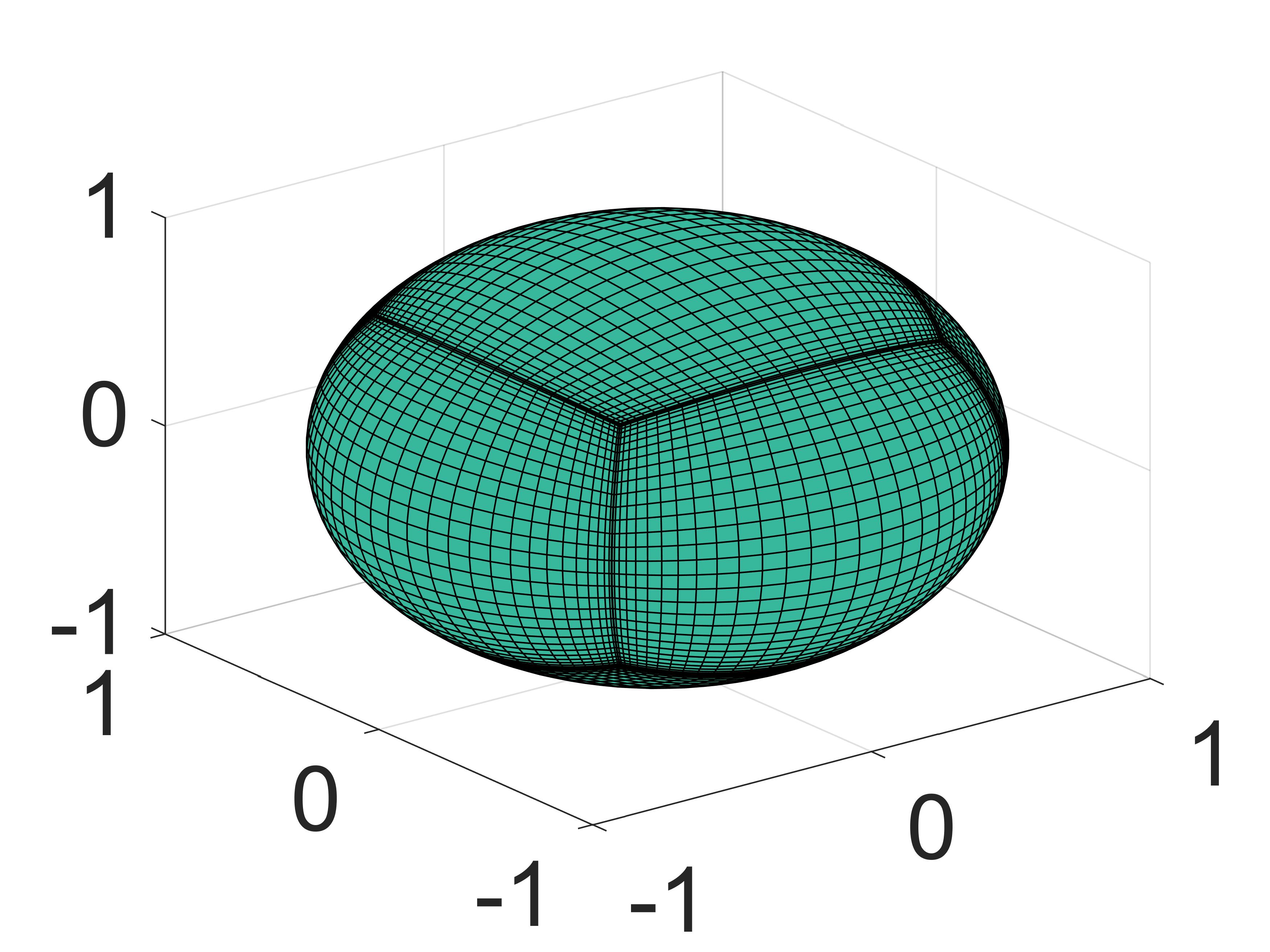} &
\includegraphics[height=4.5cm,width=4.5cm]{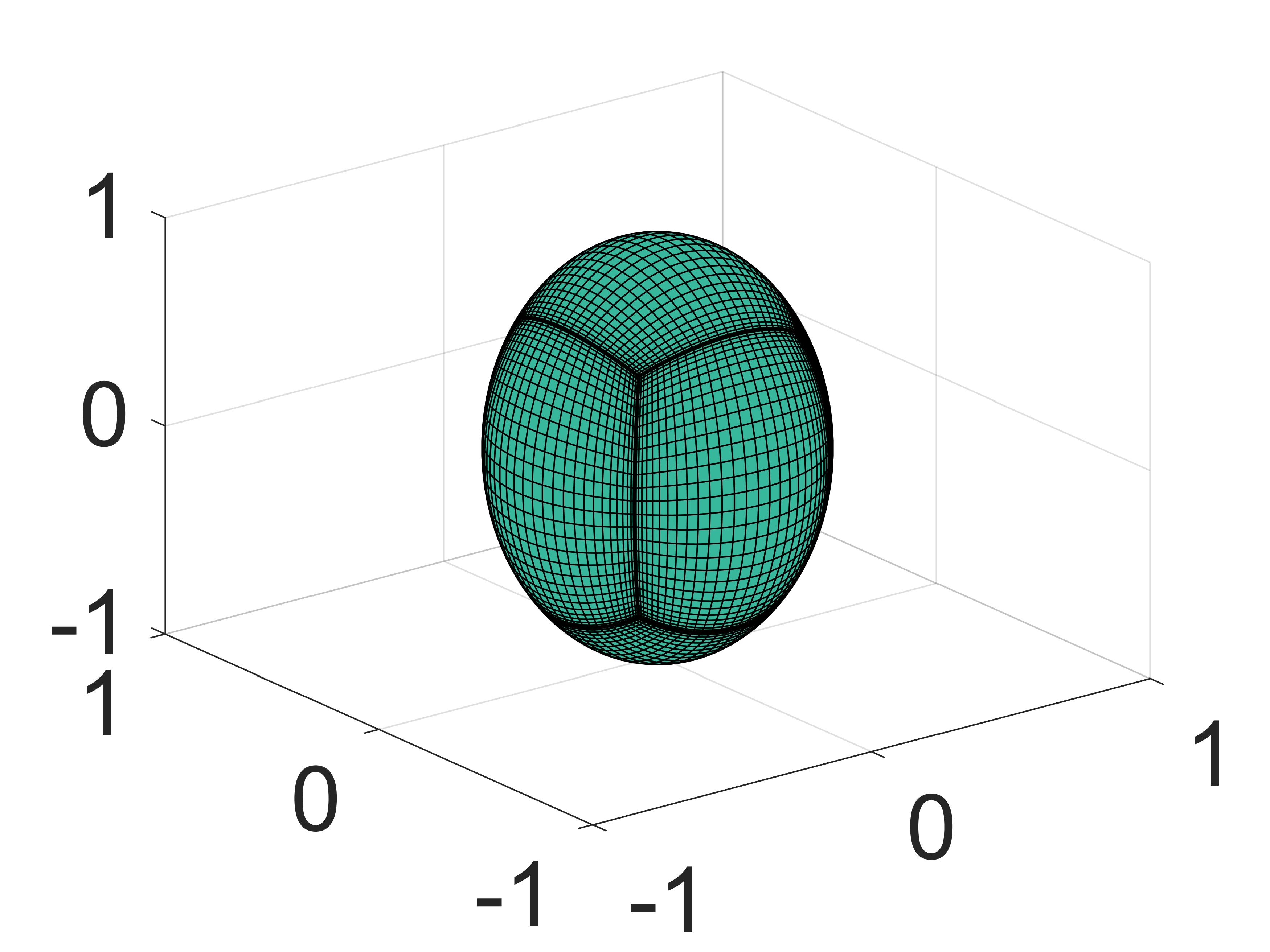} &
\includegraphics[height=4.5cm,width=4.5cm]{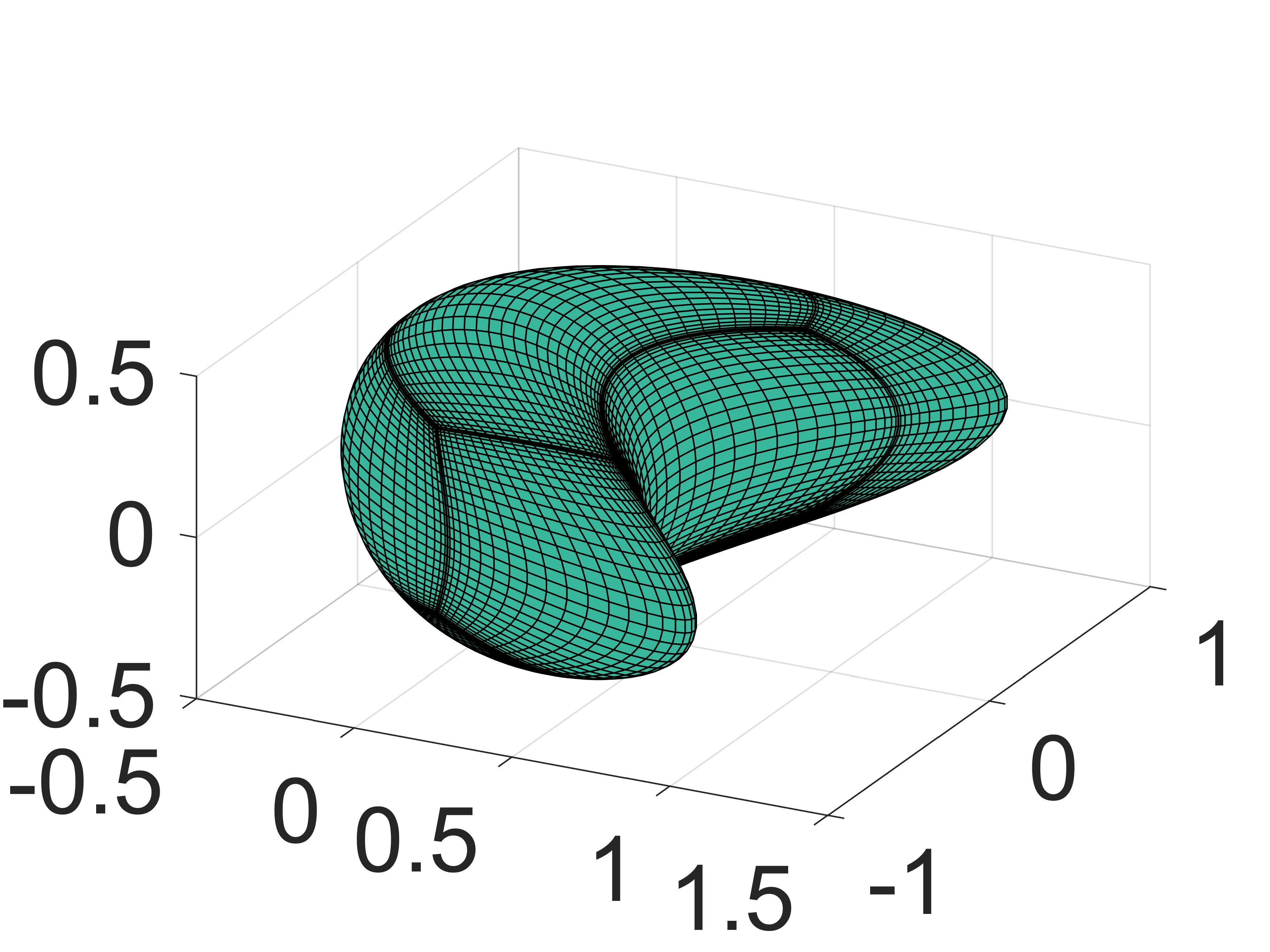} \\
(a) Ball & (b) Ellipsoid & (c) Bean \\
\end{tabular}
\caption{Obstacles used in the numerical tests presented in
  Section~\ref{sec:5}.}
\label{obstacle}
\end{figure}

\section{Numerical experiments}
\label{sec:5}
This section presents a variety of numerical tests that demonstrate
the accuracy and efficiency of the proposed three-dimensional elastic
scattering solver---or, more precisely, the accuracy and efficiency of
the computational implementations presented in Section~\ref{sec:4} for
the regularized integral equations (\ref{BIER1}) and (\ref{BIER2}) and
associated field evaluation expressions. For definiteness, the Lam\'e
constants and densities for the elastic medium are assumed as follows:
$\lambda=2$, $\mu=1$, $\rho=1$. Solutions for the integral equations
were produced by means of the fully complex version of the iterative
solver GMRES with residual tolerance $\epsilon_r$ as specified in each
case. The maximum errors presented in this section are calculated in
accordance with the expression \ben \epsilon_\infty:= \frac{\max_{x\in
    S}\{|u^\mathrm{num}(x)-u^\mathrm{ref}(x)|\}}{\max_{x\in
    S}\{|u^\mathrm{ref}(x)|\}}, \enn where $S$ is the square
$[-1,1]\times [-1,1]\times\{2\}\subset D$, and where $u^\mathrm{ref}$
is produced, for each example, through evaluation of exact solutions
$u^\mathrm{ex}$ when available, or by means of numerical solution with
sufficiently fine discretizations, otherwise. All of the numerical
tests were obtained by means of Fortran numerical implementations,
parallelized using OpenMP, on a single node (twenty-four computing
cores) of a dual socket Dell R420 with two Intel Xenon E5-2670 v3 2.3
GHz, 128GB of RAM.

In our first experiment we evaluate the accuracy of the discretization
methods used for the operators $\mathcal{T}_1$ and $\mathcal{T}_2$ on
a sphere partitioned as indicated in Figure \ref{obstacle}(a), and
using the scalar and vector functions \ben
v(x)=\sin(x_1)e^{i(x_2+x_3)},\quad
V(x)=(\sin(x_1),\cos(x_2),e^{ix_3}).  \enn As indicated in Section
\ref{sec:4.3}, the functions $\mathcal{T}_1v$ and $\mathcal{T}_2V$ are
evaluated in our context via term-by-term differentiation of the
Chebyshev expansions of $v$ and $V$. The resulting differentiation
errors, evaluated as a maximum over all discretization points, are
presented in Table \ref{compSG}---which, in particular, displays the
expected exponential convergence.

\begin{table}[htb]
  \caption{Errors in the evaluation of the operators $\mathcal{T}_1$ and $\mathcal{T}_2$.}
\centering
\begin{tabular}{|c|c|c|}
\hline
$N$ & $\mathcal{T}_1v$ & $\mathcal{T}_2V$ \\
\hline
5 & 1.03E-1 & 1.83E-2 \\
\hline
10 & 3.22E-4 & 6.36E-5 \\
\hline
15 & 1.29E-6 & 7.41E-8 \\
\hline
20 & 1.82E-9 & 9.96E-11 \\
\hline
25 & 3.16E-12 & 2.11E-12 \\
\hline
\end{tabular}
\label{compSG}
\end{table}

\begin{figure}[htbp]
\centering
\begin{tabular}{cc}
\includegraphics[scale=0.05]{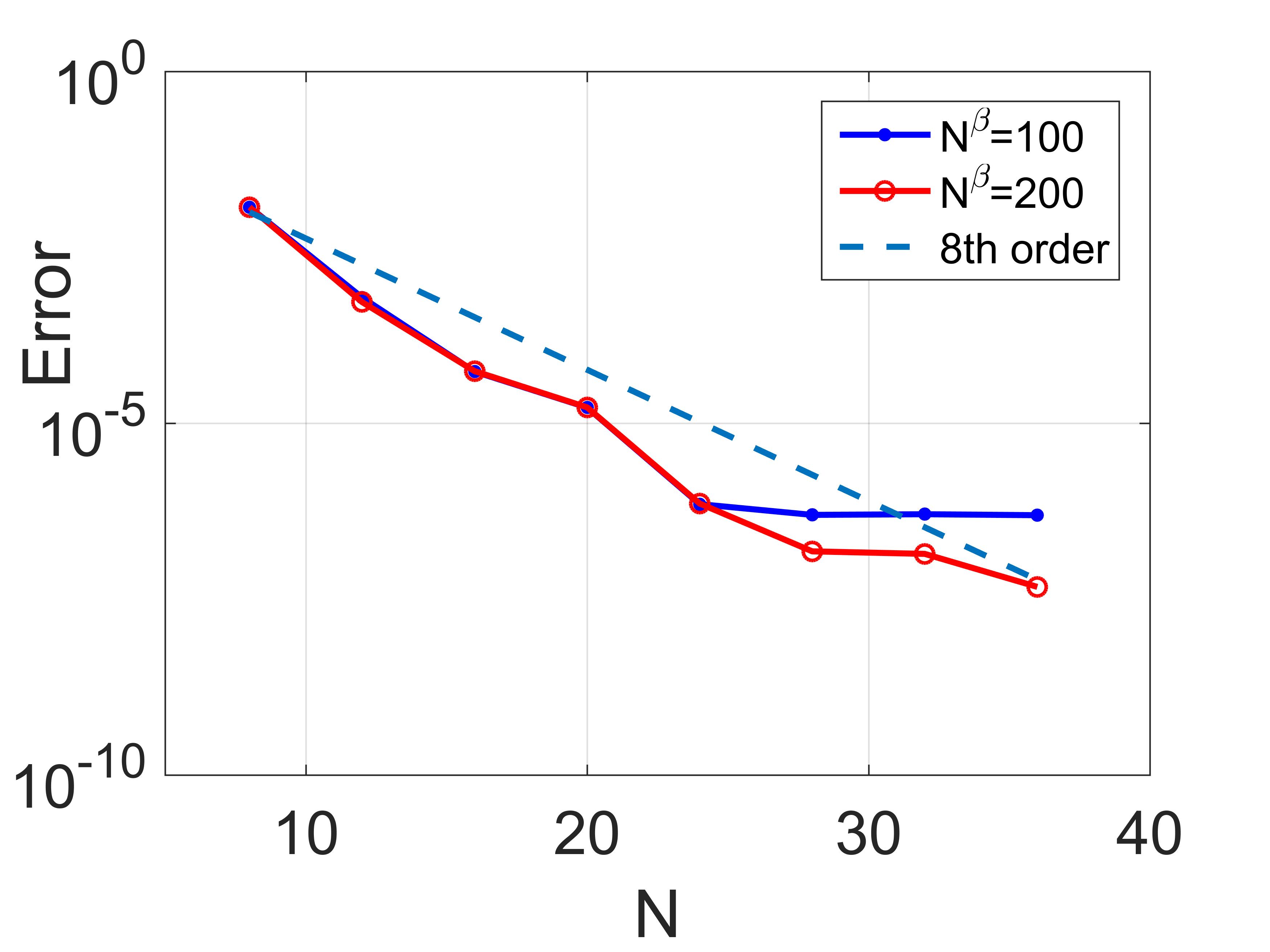} &
\includegraphics[scale=0.05]{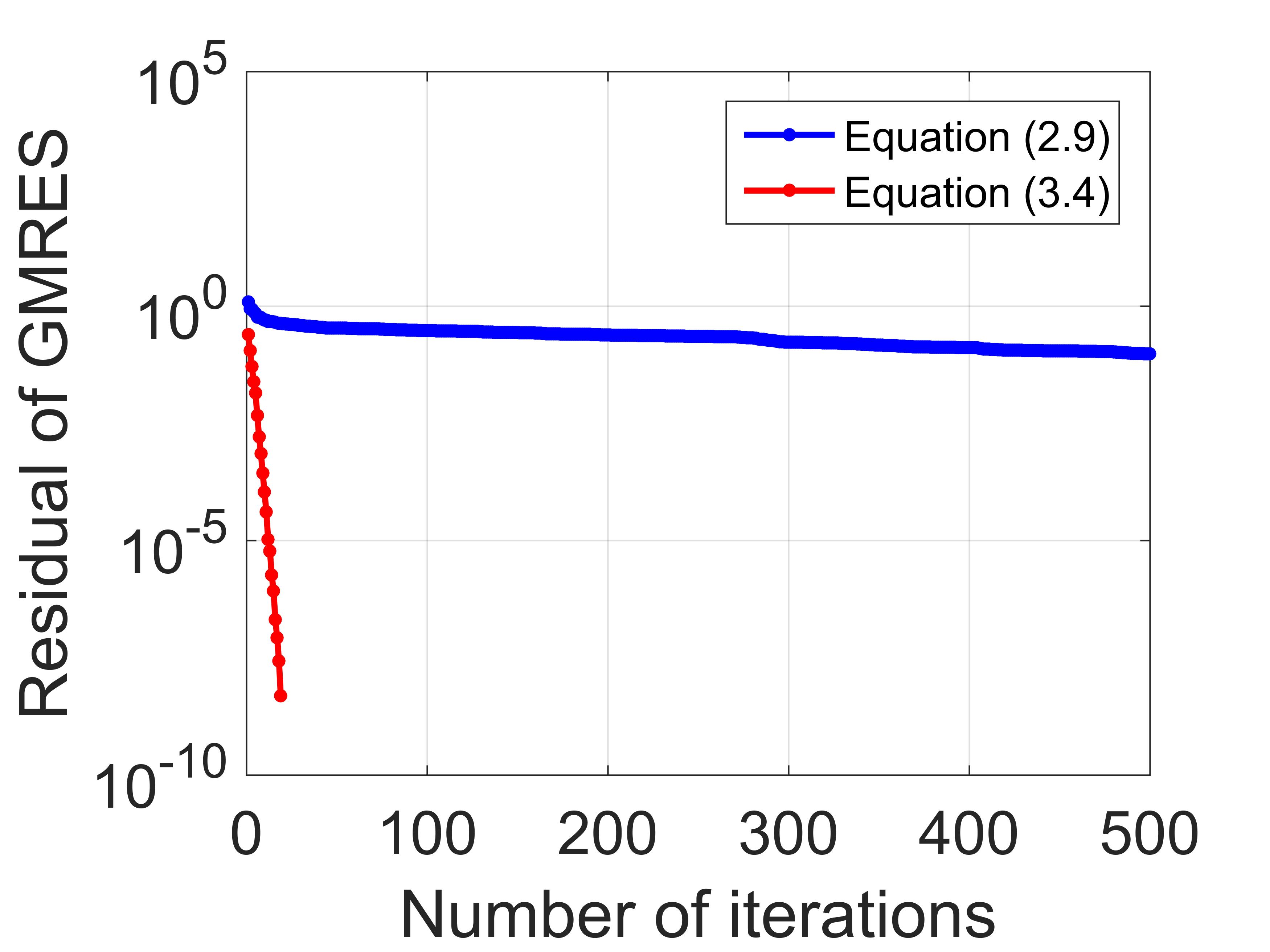} \\
(a) $\epsilon_\infty$ & (b) $\epsilon_r$\\
\end{tabular}
\caption{Numerical errors (a) and GMRES residual (b) for the problem
  of scattering by the spherical obstacle.}
\label{FigExp2.1}
\end{figure}

\begin{figure}[htbp]
\centering
\begin{tabular}{cc}
\includegraphics[scale=0.05]{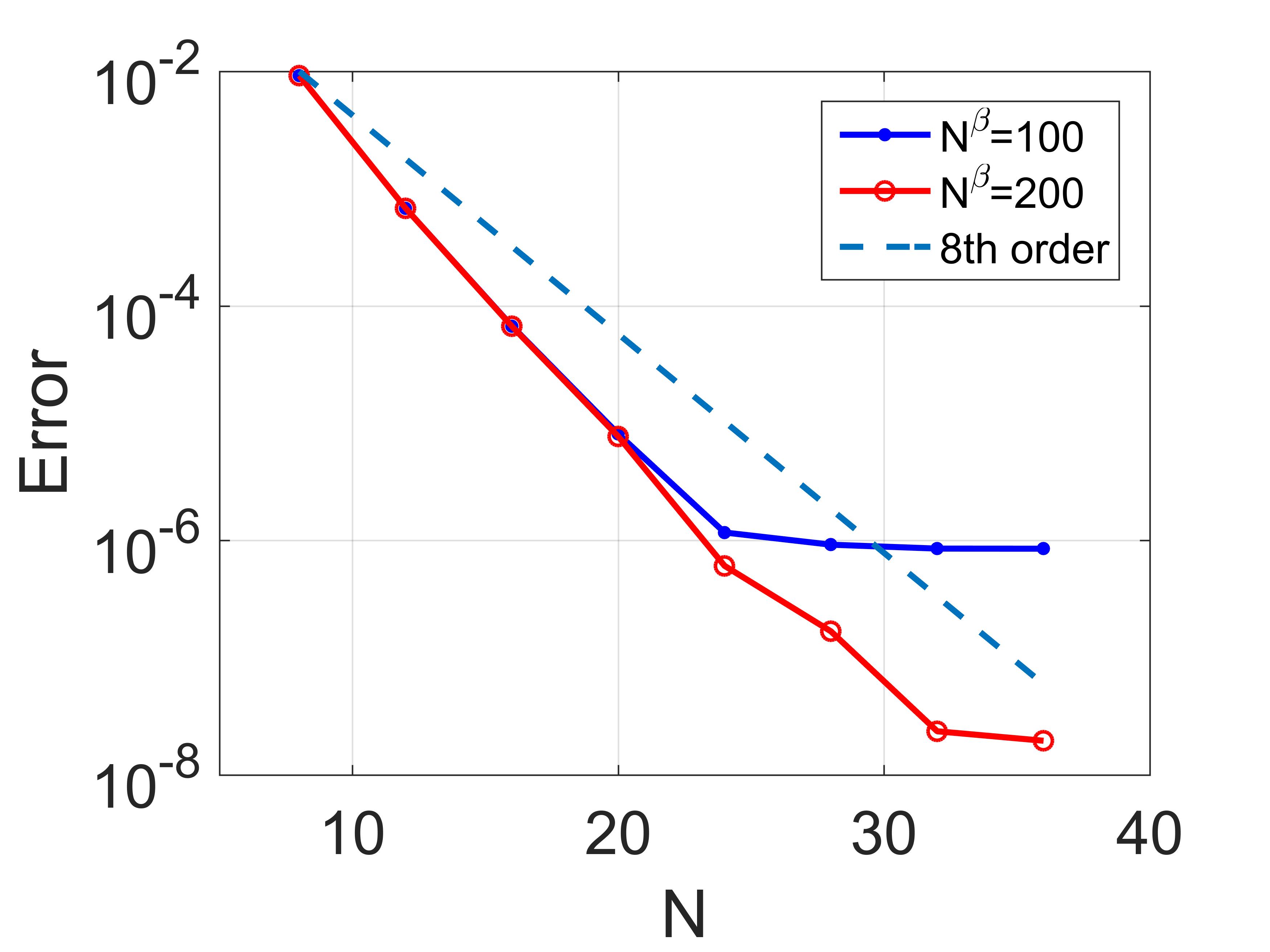} &
\includegraphics[scale=0.05]{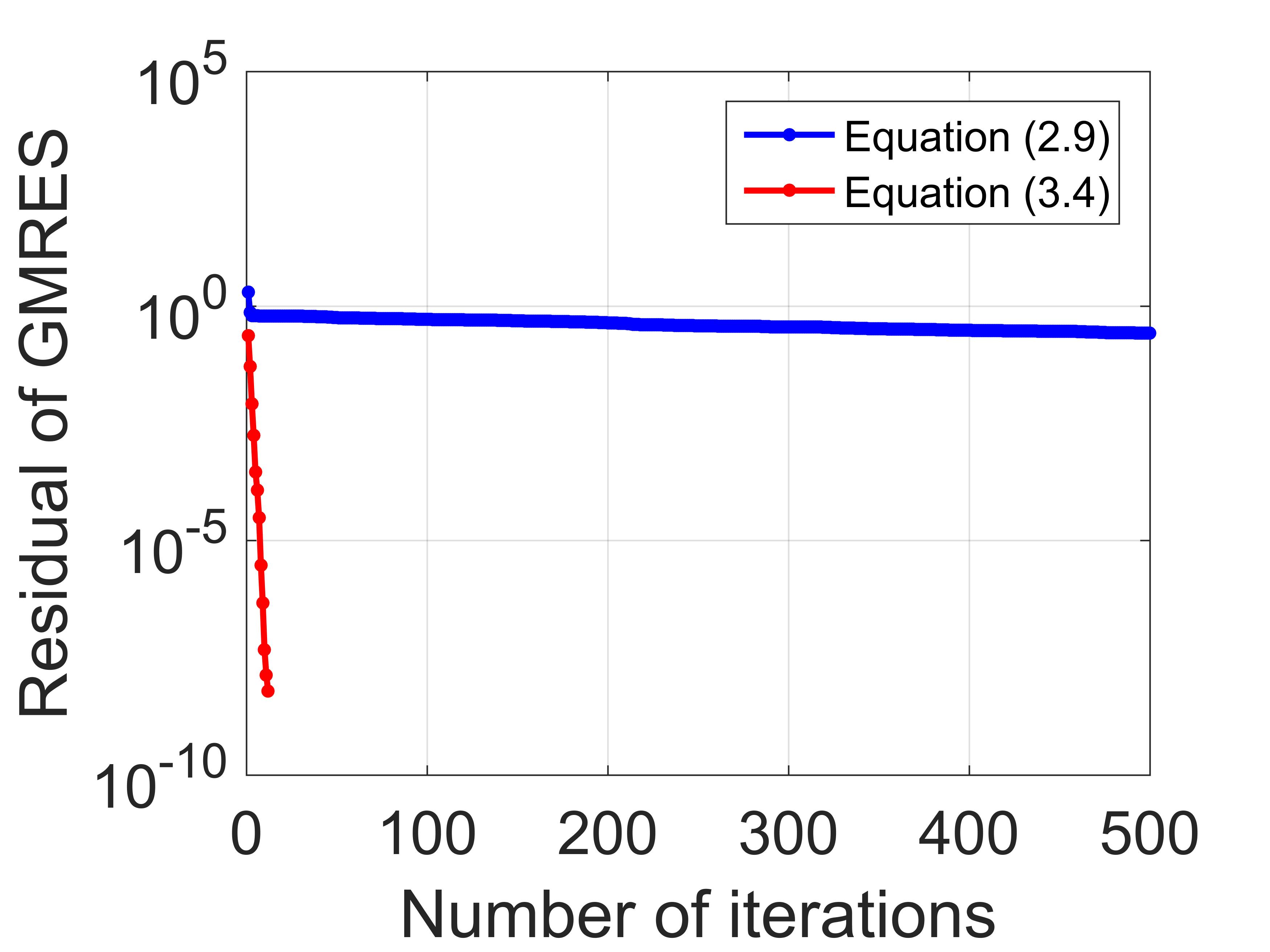} \\
(a) $\epsilon_\infty$ & (b) $\epsilon_r$\\
\end{tabular}
\caption{Numerical errors (a) and GMRES residual (b) for the problem
  of scattering by the ellipsoidal obstacle.}
\label{FigExp2.2}
\end{figure}

\begin{figure}[htbp]
\centering
\begin{tabular}{cc}
\includegraphics[scale=0.05]{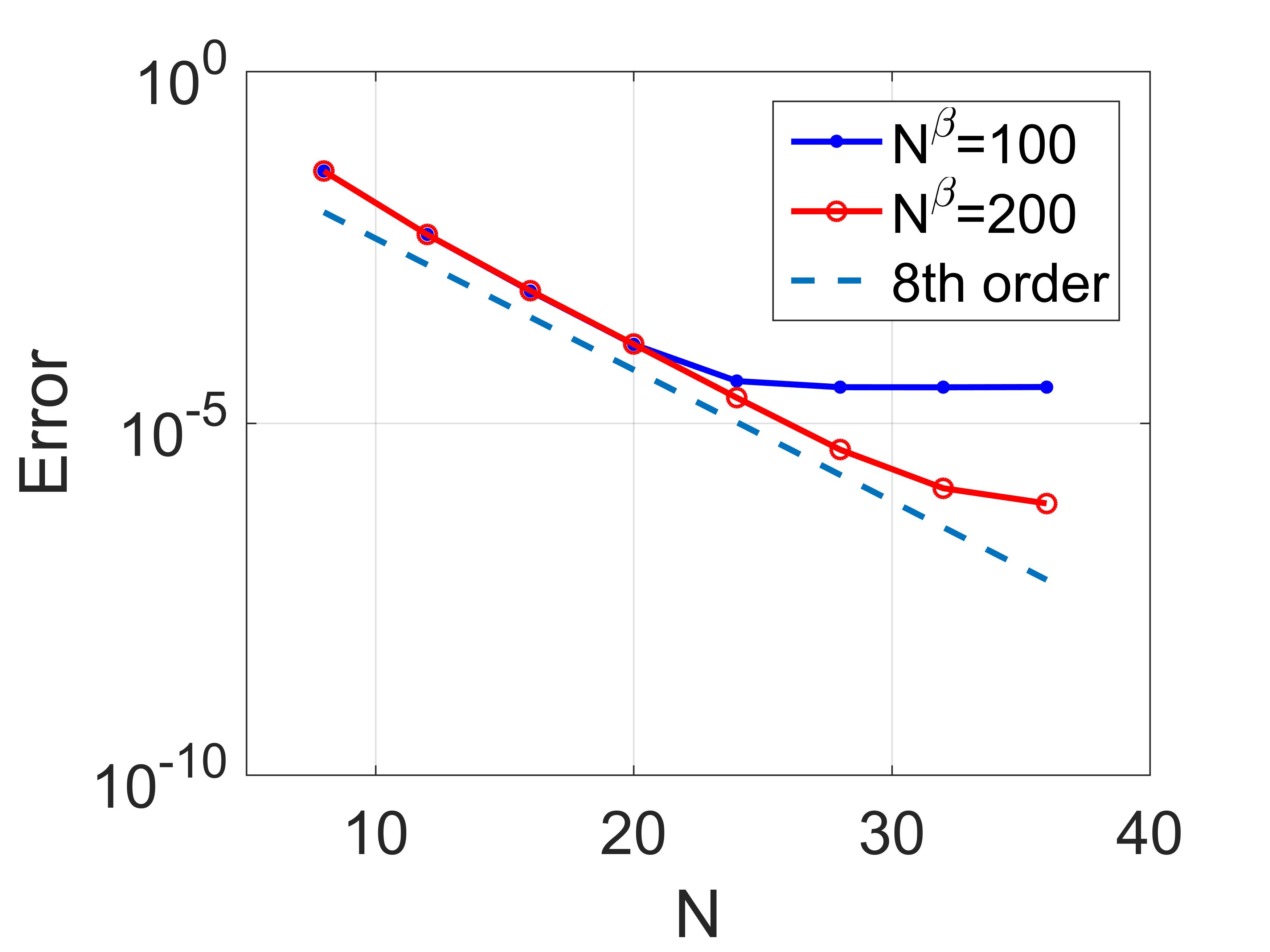} &
\includegraphics[scale=0.05]{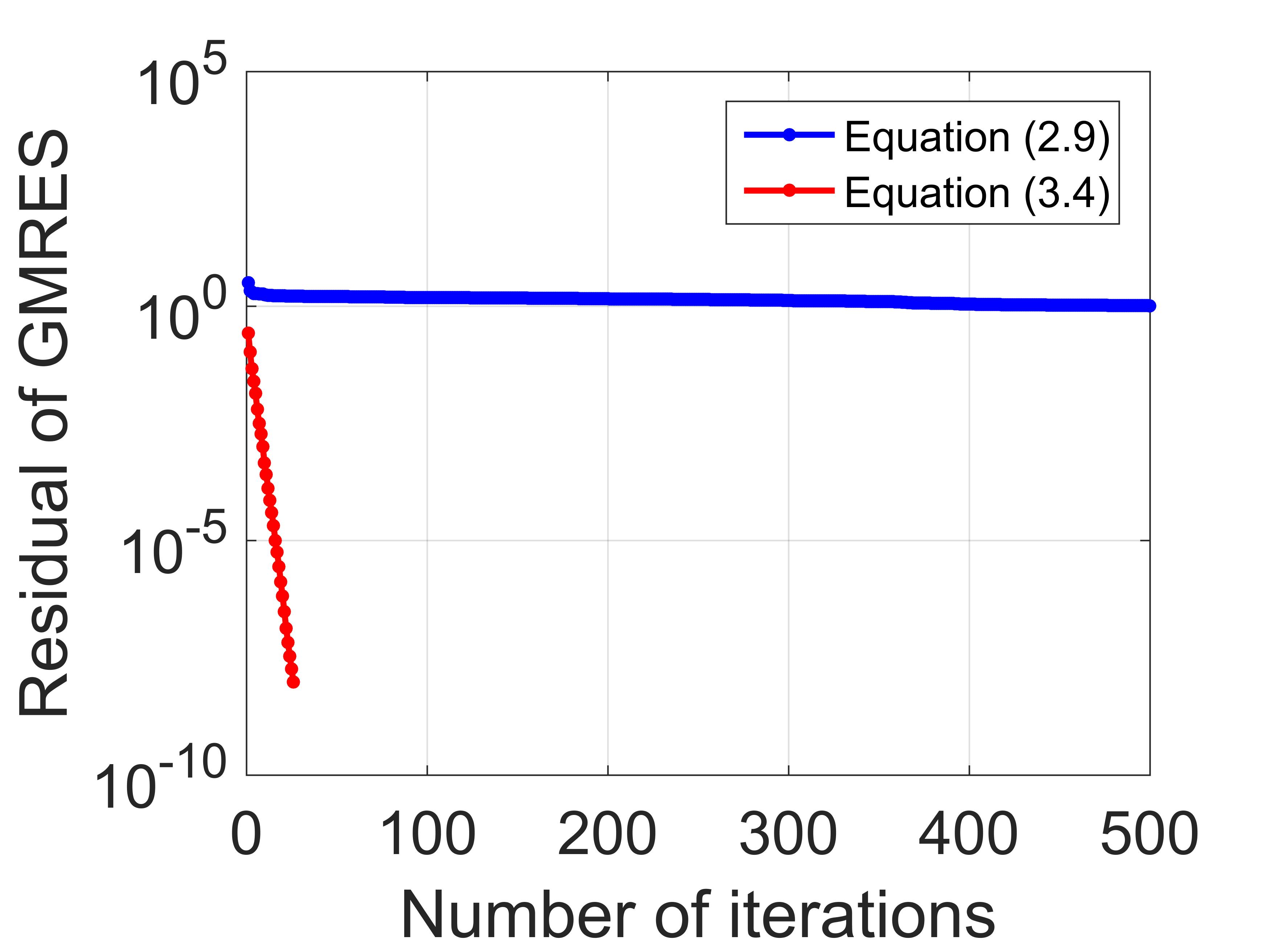} \\
(a) $\epsilon_\infty$ & (b) $\epsilon_r$\\
\end{tabular}
\caption{Numerical errors (a) and GMRES residual (b) for the problem
  of scattering by the bean-shaped obstacle.}
\label{FigExp2.3}
\end{figure}

Next, we demonstrate the high accuracy and rapid convergence of the
proposed closed-surface elastic scattering method via applications to
the three bounded obstacles depicted in Figure~\ref{obstacle}. In each
case boundary conditions were used for which the exact solution is
given by an pressure point source located at a point $z$ within the
ball: \ben
u^\mathrm{ex}(x)=\frac{1}{k_p}\nabla_x\frac{e^{ik_p|x-z|}}{4\pi|x-z|}.
\enn (While not physically motivated, this exact solution and
associated boundary conditions provide a commonly used test for
evaluation of the accuracy of the scattering solver.) The source was
assumed to be located at $z=(0,0.5,0.3)$ for the spherical scatterer,
and at $z=(0,0,0)$ for the ellipsoidal and bean-shaped
obstacles. Figures \ref{FigExp2.1}(a), \ref{FigExp2.2}(a) and
\ref{FigExp2.3}(a) display the errors in the numerical solution for
the frequency $\omega=2\pi$, produced by means of the regularized
integral equation (\ref{BIER1}), as a function of $N$. In all three
cases $M=6$ patches were used, together with two different values of
the rectangular-integration parameter, namely $N^\beta = 100$ and
$N^\beta = 200$. These figures clearly demonstrate the fast
convergence and high accuracy of the algorithm. Figures
\ref{FigExp2.1}(b), \ref{FigExp2.2}(b) and \ref{FigExp2.3}(b), in
turn, display the GMRES residuals as functions of the number of
iterations, in the numerical solution of the un-regularized
(resp. regularized) integral equation~(\ref{BIE1})
(resp. (\ref{BIER1})), for which we used $N=36$ and
$N^\beta=200$. Clearly, use of the regularized equation is highly
beneficial: using only 19, 12 and 28 iterations the solver achieves
the GMRES tolerance $\epsilon_r = 1\times10^{-8}$ for the spherical,
ellipsoidal and bean-shaped obstacles, respectively. This is in
striking contrast with the numbers of iterations required by the
implementation based on the unregularized equation, which are also
displayed in these figures. Table \ref{TableExp2.1} presents the
numerical solution errors together with other statistics such as
precomputation time, time per iteration and number of iterations used
for a problem of scattering at frequency $\omega=10\pi$ on the basis
of six $5\times5$ patches. At this frequency, $N=8$ (resp. $N=16$)
suffices to produce an accuracy $4.67\times10^{-3}$
(resp. $1.39\times10^{-6}$).

\begin{table}[htb]
  \caption{Numerical errors in the numerical total field for the problem of scattering by a sphere of diameter $10\lambda_s$ produced by the solver based on the regularized equation (\ref{BIER1}) .}
\centering
\begin{tabular}{|c|c|c|c|c|c|c|c|}
\hline
$N$ & $N^\beta$ & $N_{\mathrm{DOF}}$  & Time (prec.) & Time (1 iter.)& $N_{iter}$ ($\epsilon_r$) & $\epsilon_\infty$\\
\hline
 8  & 50  & $3\times9,600$  & 11.63 s & 8.32 s & 22 ($9.41\times10^{-4}$) & $4.77\times10^{-3}$ \\
\hline
 8  & 100 & $3\times9,600$  & 43.22 s & 8.32 s & 21 ($9.88\times10^{-4}$) & $4.67\times10^{-3}$ \\
\hline
16  & 50  & $3\times38,400$ & 1.09 min & 2.14 min & 32 ($1.05\times10^{-5}$) & $1.32\times10^{-4}$ \\
\hline
16  & 100 & $3\times38,400$ & 3.52 min & 2.16 min & 34 ($9.73\times10^{-7}$) & $1.39\times10^{-6}$ \\
\hline
\end{tabular}
\label{TableExp2.1}
\end{table}

We next consider the plane-wave incident pressure field \be
\label{pplane}
u^{inc}=de^{ik_px\cdot d}, \quad
d=(\sin\theta_1\cos\theta_2,\sin\theta_1\sin\theta_2, \cos\theta_1),
\en where $(\theta_1,\theta_2)$ denote the polar and azimuthal
incidence angles. For our example we use the two pairs of angles
$\theta_1=\pi/2$, $\theta_2=0$ and $\theta_1=\pi$, $\theta_2=0$ for
the spherical, and bean-shaped obstacles, respectively, and we take
$\omega=10\pi$, $M=6\times 5\times 5 =150$ patches (the six original
patches subdivided into $5\times5$ each), and $N=16$ (for a total
number $N_{\mathrm{DOF}}=115200$ of degrees of freedom in the
problem). Figures \ref{FigExp2.5} and \ref{FigExp2.6} display the
resulting numerical solutions.

\begin{figure}[ht]
\centering
\begin{tabular}{ccc}
\includegraphics[scale=0.05]{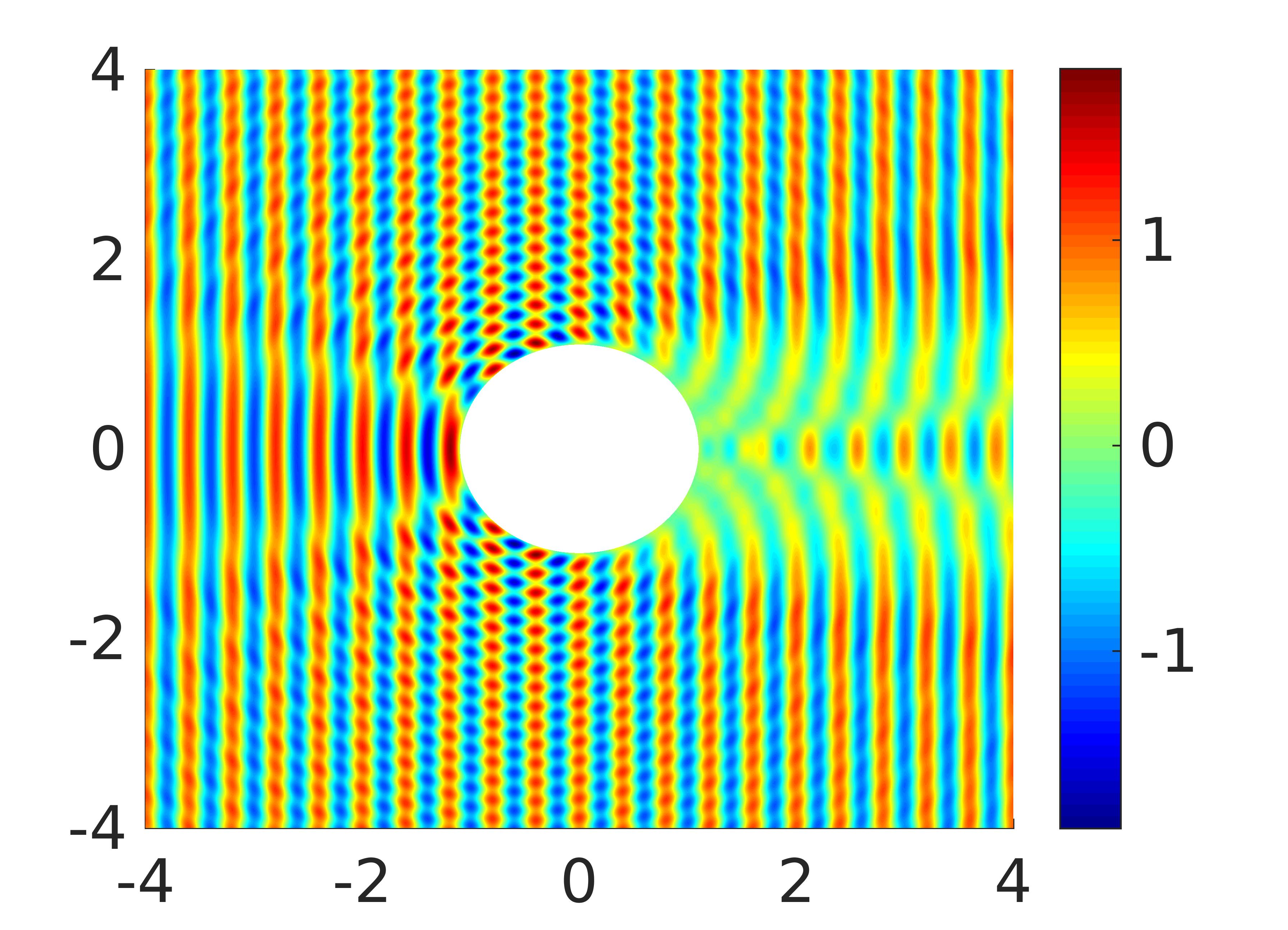} &
\includegraphics[scale=0.05]{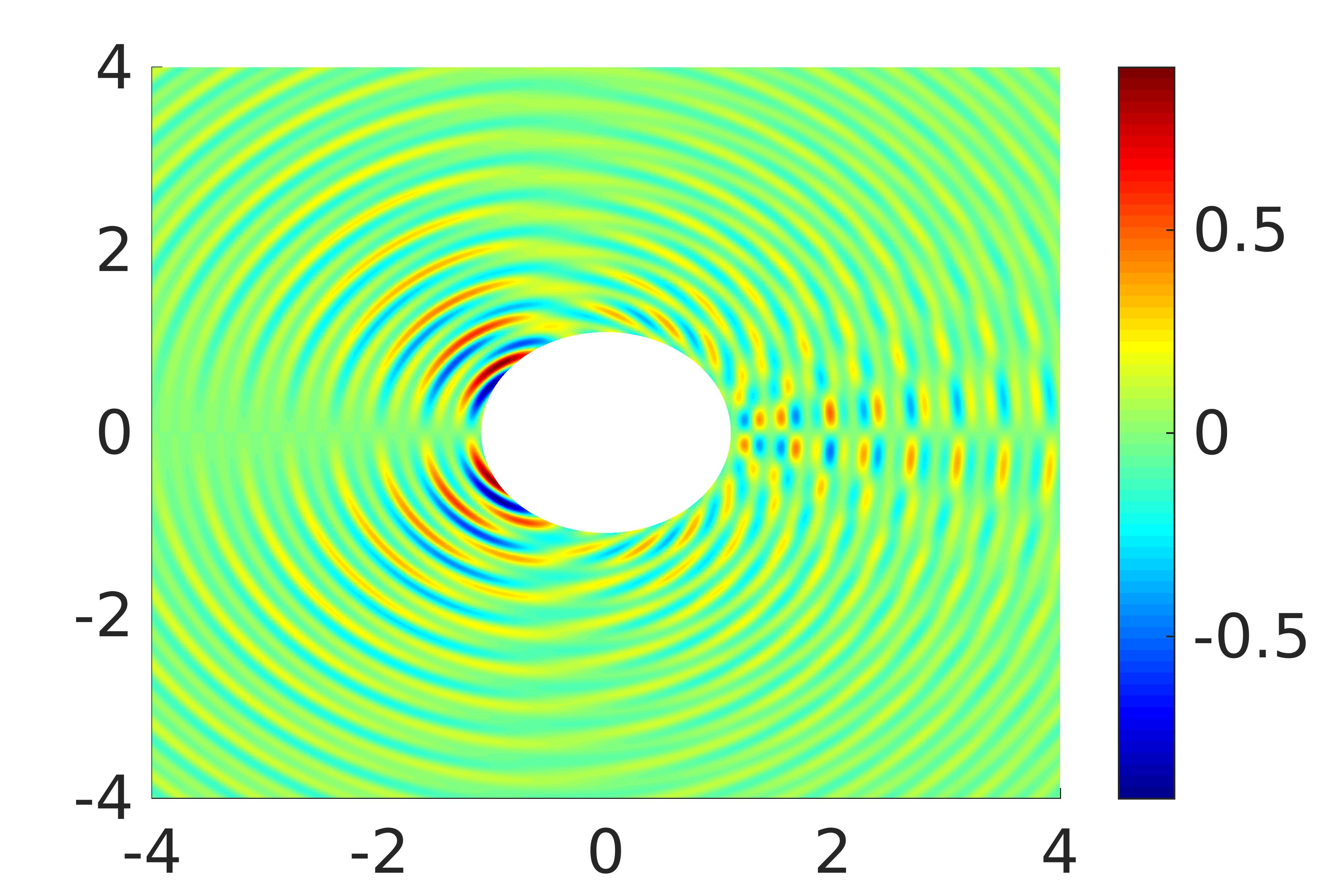} \\
(a) $\mbox{Re}(u^{\mathrm{num}}_1)$ & (b) $\mbox{Re}(u^{\mathrm{num}}_2)$ \\
\includegraphics[scale=0.05]{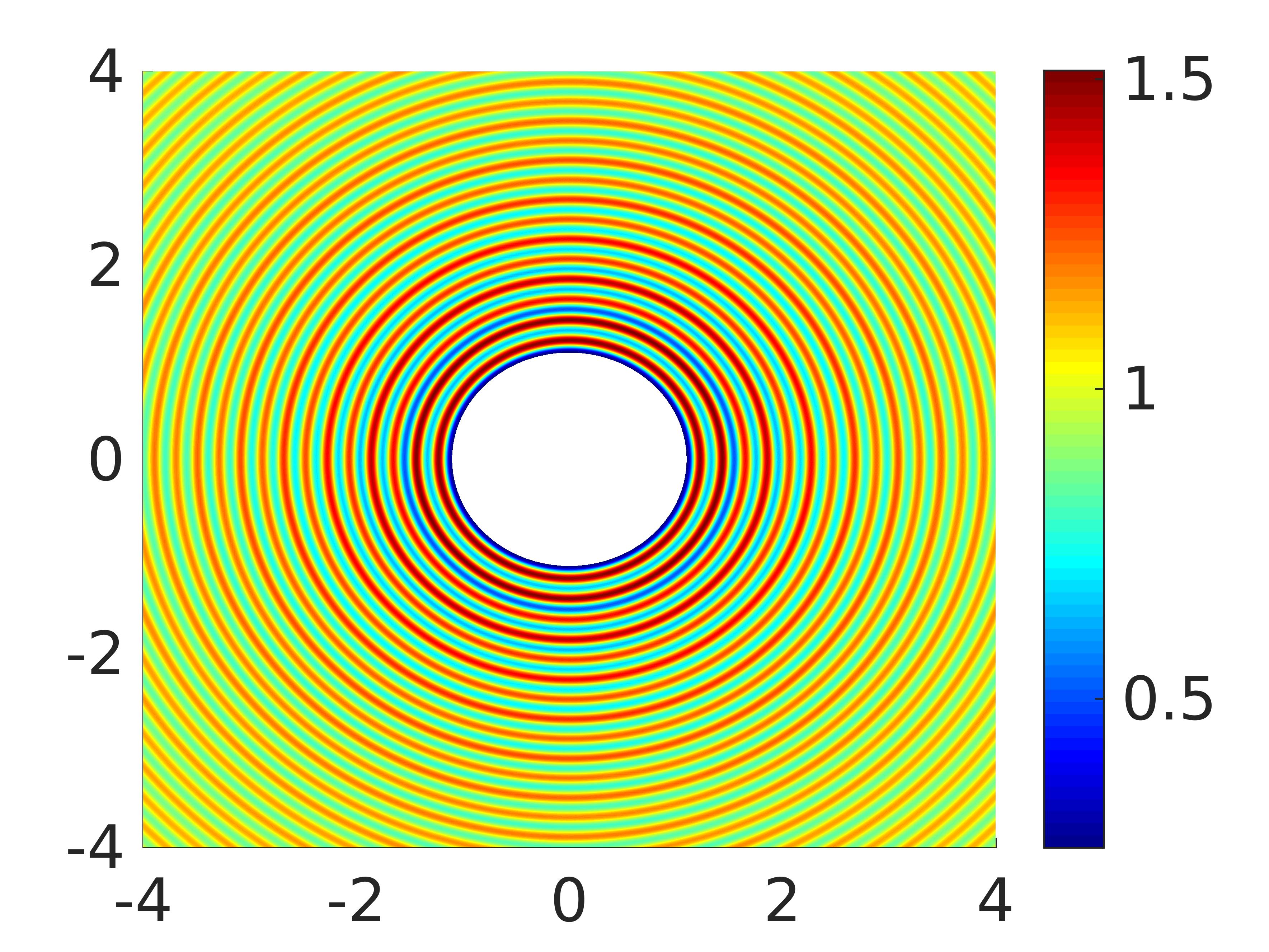} &
\includegraphics[scale=0.05]{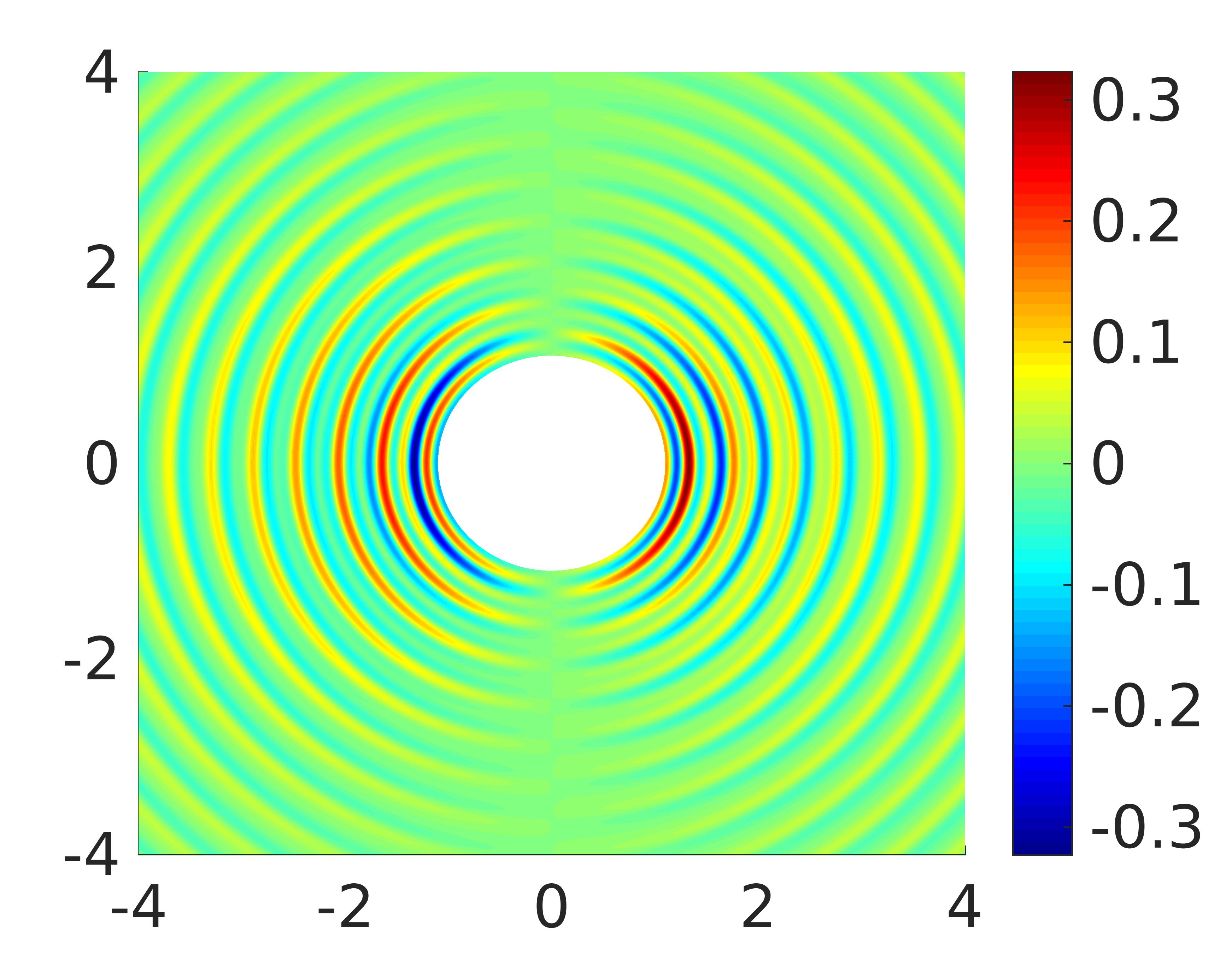} \\
(c) $\mbox{Re}(u^{\mathrm{num}}_1)$ & (d) $\mbox{Re}(u^{\mathrm{num}}_2)$ \\
\end{tabular}
\caption{Real parts of the components $u_1$ and $u_2$ of the total
  field $u$ on an $x_3=0$ section (Figs. (a) and (b)) and an $x_1=0$
  section (Figs. (c) and (d)) for the scattering of a plane-wave
  pressure incident field, with incidence angles $\theta_1=\pi/2$,
  $\theta_2=0$, by the spherical obstacle. A total of forty-seven
  iterations sufficed in this case for the solver to reach the GMRES
  residual tolerance value $\epsilon_r=1\times10^{-4}$.}
\label{FigExp2.5}
\end{figure}

\begin{figure}[ht]
\centering
\includegraphics[scale=0.06]{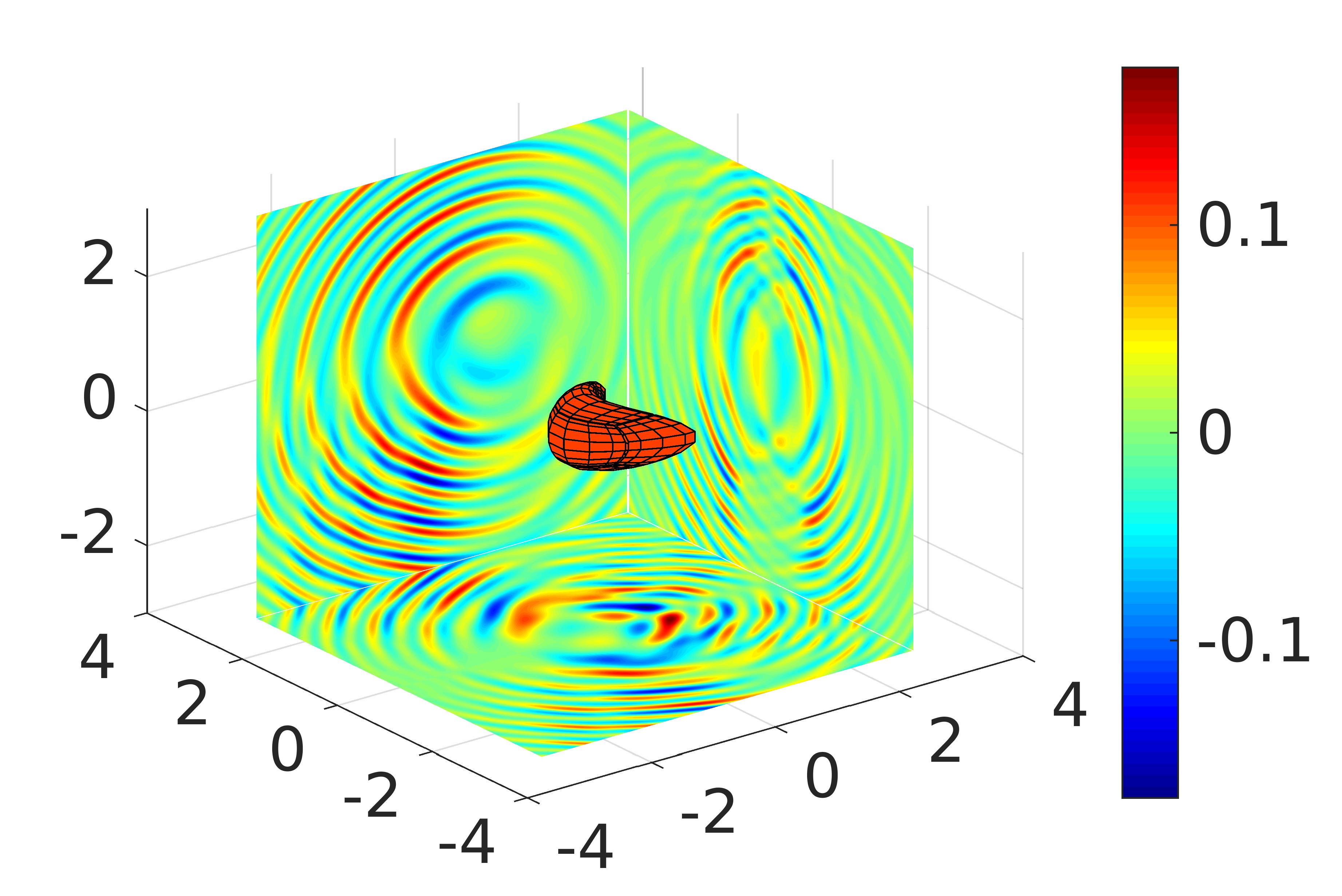} \\
(a) $\mathrm{Re}(u^{\mathrm{num}}_1)$ \\
\begin{tabular}{ccc}
\includegraphics[scale=0.06]{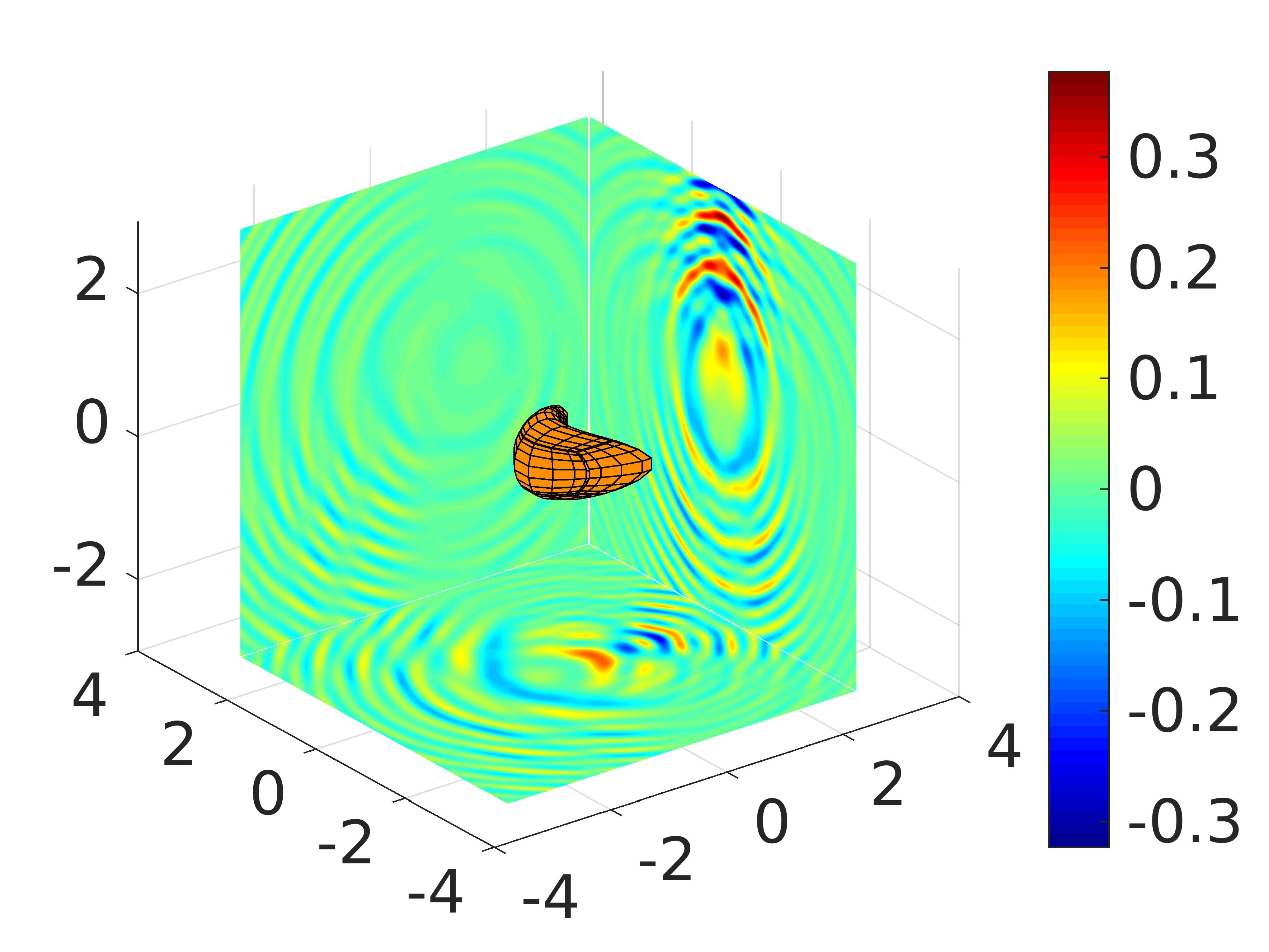} &
\includegraphics[scale=0.06]{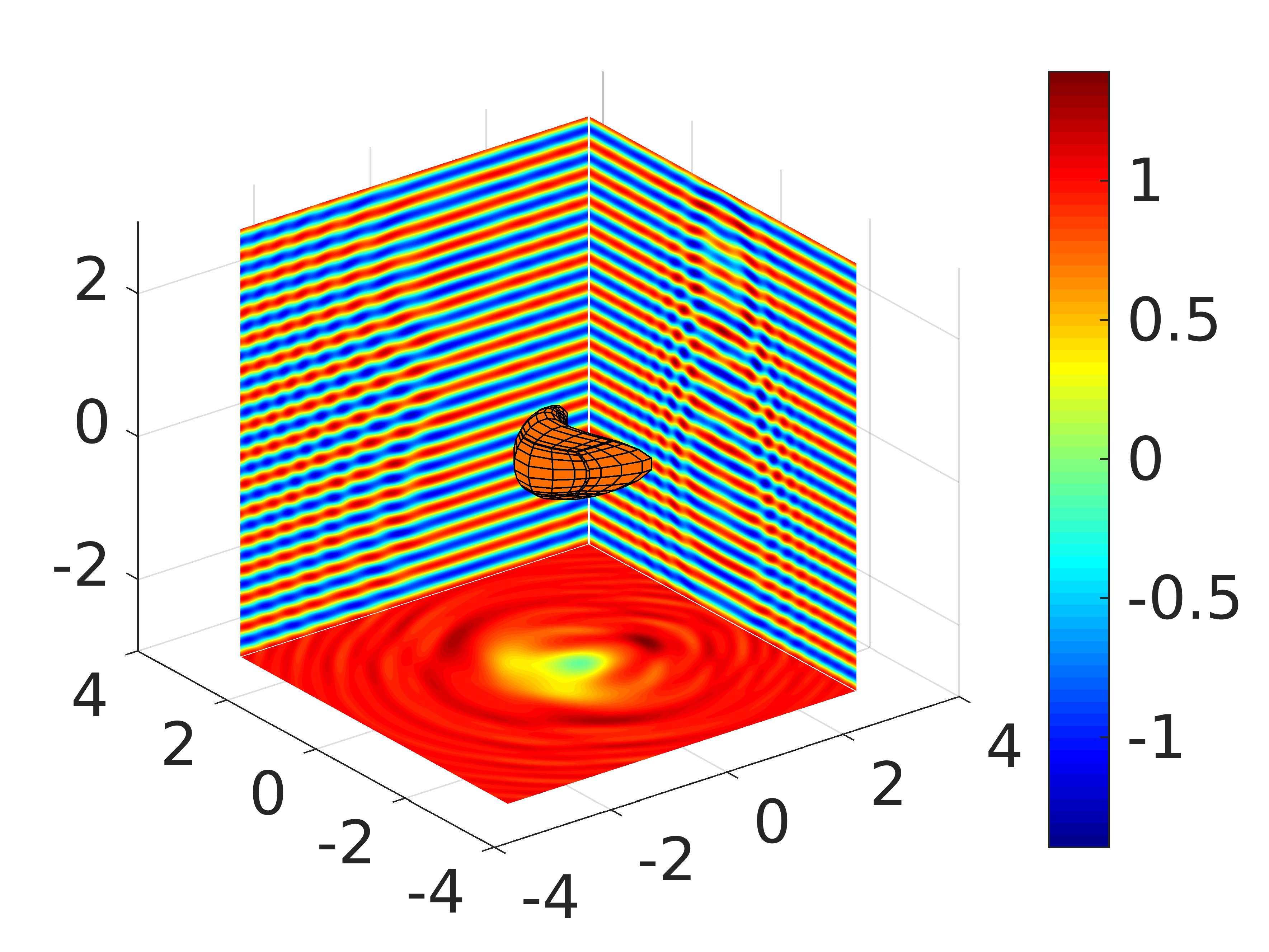} \\
(b) $\mathrm{Re}(u^{\mathrm{num}}_2)$ & (c) $\mathrm{Re}(u^{\mathrm{num}}_3)$\\
\end{tabular}
\caption{Scattering of a plane-wave pressure incident field with
  incident angles $\theta_1=\pi$ and $\theta_2=0 $ by the bean-shaped
  obstacle. Total field. The solver required eighty-one iterations to
  reach the GMRES tolerance value $\epsilon_r= 1\times10^{-4}$.}
\label{FigExp2.6}
\end{figure}


Finally, we consider the problem of elastic scattering by a unit disc
\ben x^2+y^2\le 1,\quad z=0, \enn for which the weight function
$w=\sqrt{1-x^2-y^2}$ was used, under plane pressure incidence
field~(\ref{pplane}) with incidence angles $\theta_1=\pi$ and
$\theta_2=0$. This problem can be tackled by means of either the
first-kind equation (\ref{BIE2}) or the regularized equation
(\ref{BIER2}). Figure \ref{FigExp3.1} displays the errors in the
numerical solution as a function of $N$, demonstrating once again fast
convergence and high accuracy. The values $M=5$ and $N^\beta=200$ were
used. Figures \ref{FigExp3.11} and \ref{FigExp3.12} present the
GMRES residuals as a function of the number of iterations for the
various formulations (\ref{BIE2}), (\ref{BIER2}), (\ref{BIED}) and
(\ref{BIEDR}). Clearly, the Neumann solver based on the regularized
integral equation (\ref{BIER2}) requires a significantly smaller
number of GMRES iterations, to meet a given GMRES tolerance
$\epsilon_r$, than the corresponding solver based on equation
(\ref{BIE2}). But for the Dirichlet problem, the regularized
equation~(\ref{BIEDR}) does not provide an improvement over equation
(\ref{BIED}): it actually requires a slightly larger number of
iterations in this case. The total computing cost of the regularized
equation (\ref{BIEDR}) is higher than (\ref{BIED}) in this case,
since the application of the operator $S_w$ is significantly less
expensive than the application of operator $N_wS_w$---and, thus, use
of the formulation based on the unregularized operator $S_w$ is
recommended for the Dirichlet case.  It is worth noting that, in
absolute computing times, the cost of evaluation of each open-surface
operator $N_w$ and $S_w$ is comparable, for a given overall number of
discretization points, to the cost required by the corresponding
closed-surface operators $N$ and $S$, respectively;
cf. e.g. Figure~\ref{FigExp3.1}. Figures \ref{FigExp3.2} and
\ref{FigExp3.3} display the total field scattered under the Neumann
and Dirichlet problem, respectively. In Figure \ref{FigExp3.2} the
famous Poisson spot is clearly visible at the center of the shadow
area of the third component of the field.
\begin{figure}[htbp]
\centering
\includegraphics[scale=0.25]{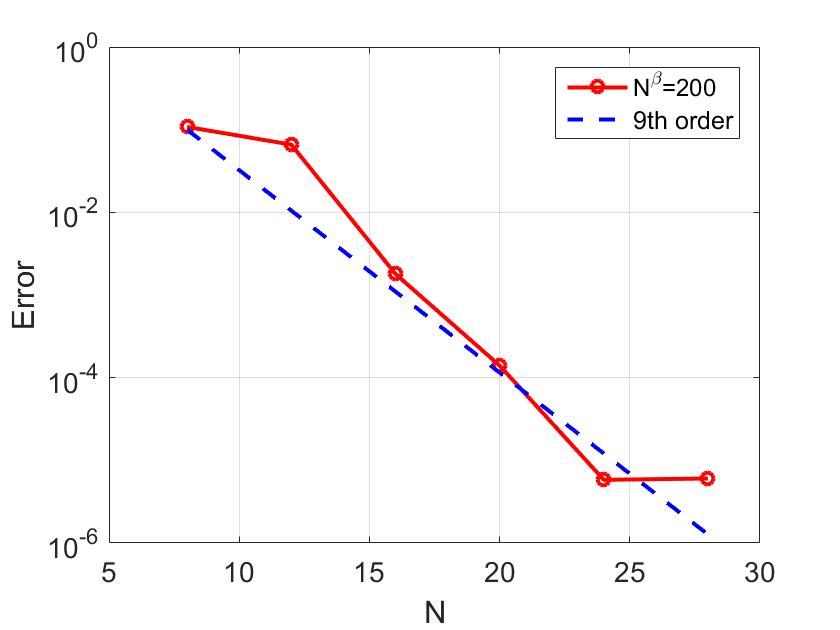}
\caption{Numerical errors for the problem of scattering by a unit disc
  of diameter $2\lambda_s$. The accuracy limitation at a level of
  approximately $10^{-5}$ corresponds to the choice $N^\beta = 200$;
  higher accuracies can be obtained by using suitably larger values of
  this precomputation-related parameter.}
\label{FigExp3.1}
\end{figure}

\begin{figure}[htbp]
\centering
\begin{tabular}{cc}
\includegraphics[scale=0.2]{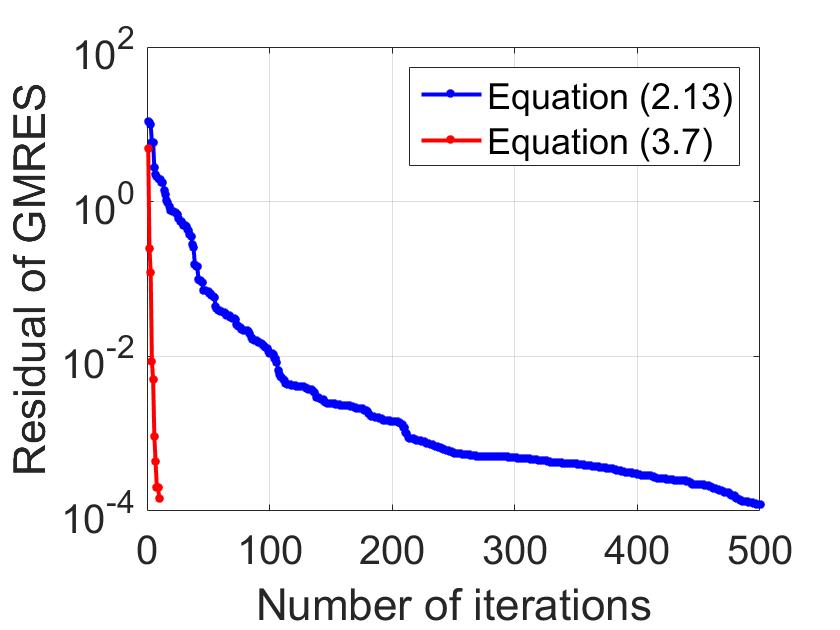} &
\includegraphics[scale=0.2]{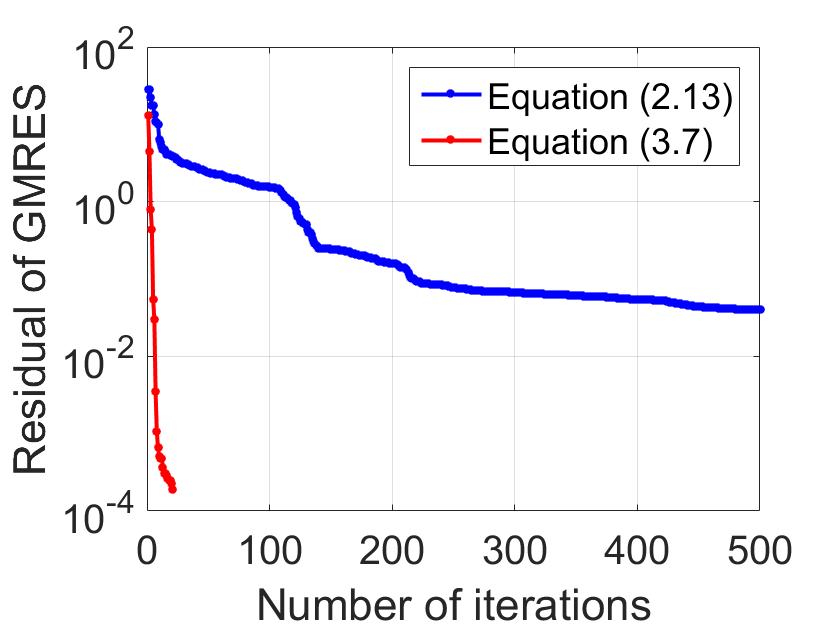} \\
(a) $N_{\mathrm{DOF}}=3\times1280$ & (b) $N_{\mathrm{DOF}}=3\times5120$ \\
\end{tabular}
\caption{GMRES residuals obtained in the solution of the Neumann
  problem of scattering by a unit disc with diameter (a) $2\lambda_s$
  and (b) $4\lambda_s$.}
\label{FigExp3.11}
\end{figure}

\begin{figure}[htbp]
\centering
\begin{tabular}{cc}
\includegraphics[scale=0.2]{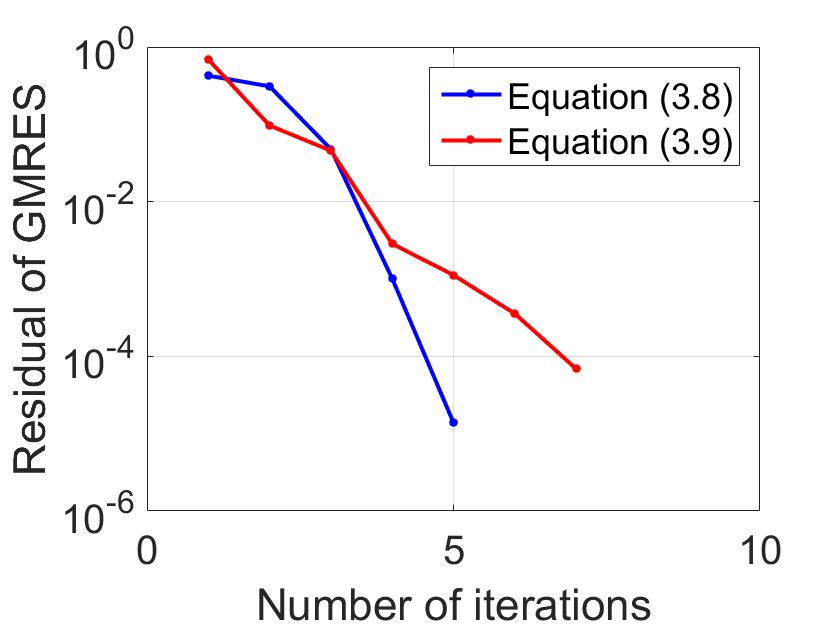} &
\includegraphics[scale=0.2]{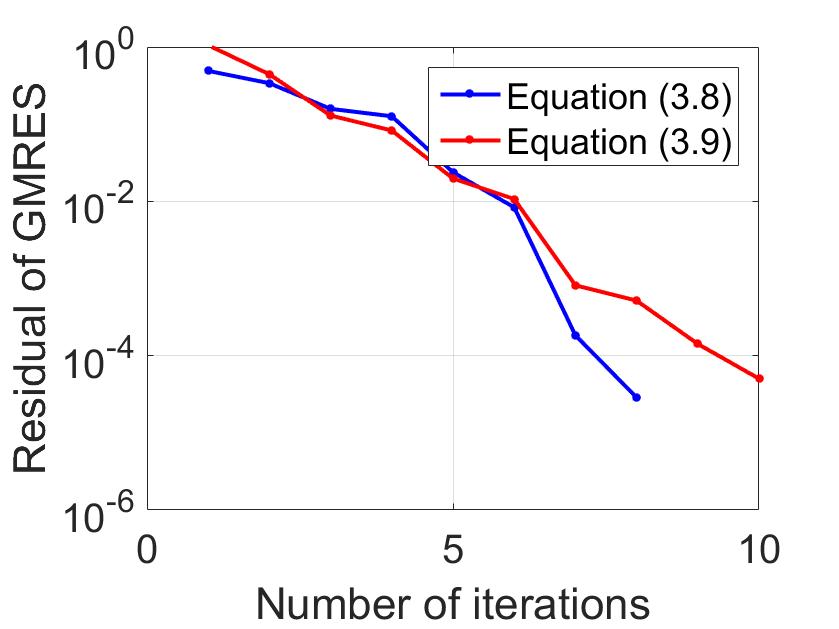} \\
(a) $N_{\mathrm{DOF}}=3\times1280$ & (b) $N_{\mathrm{DOF}}=3\times5120$ \\
\end{tabular}
\caption{GMRES residuals obtained in the solution of the Dirichlet
  problem of scattering by a unit disc with diameter (a) $2\lambda_s$
  and (b) $4\lambda_s$.}
\label{FigExp3.12}
\end{figure}

\begin{figure}[htbp]
\centering
\begin{tabular}{ccc}
\includegraphics[scale=0.3]{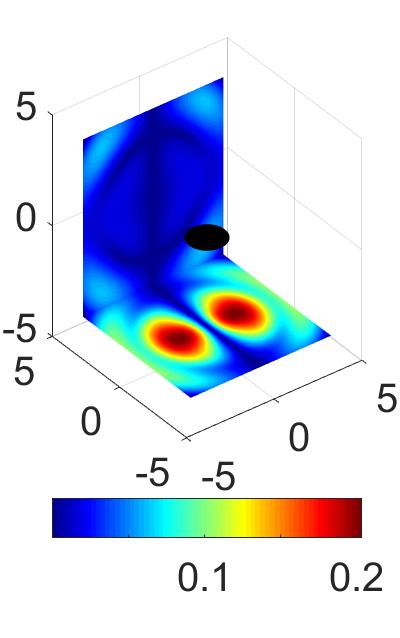} &
\includegraphics[scale=0.3]{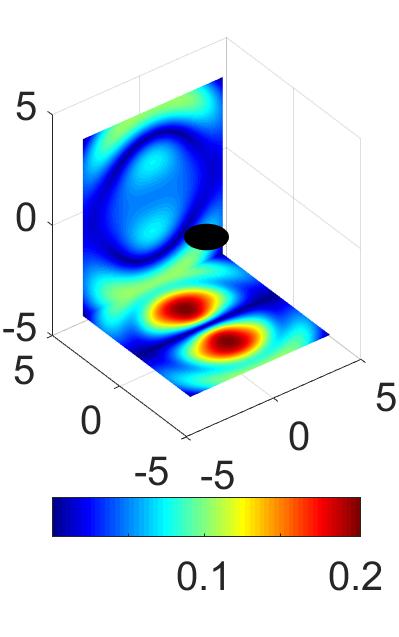} &
\includegraphics[scale=0.3]{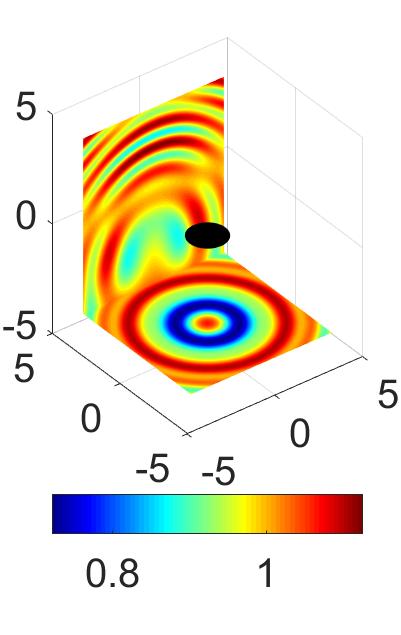} \\
(a) $|u^{\mathrm{num}}_1|$ & (b) $|u^{\mathrm{num}}_2|$ & (c) $|u^{\mathrm{num}}_3|$ \\
\end{tabular}
\caption{Scattering of a plane-wave pressure incident wave by a unit
  disc of diameter $2\lambda_s$. Total field under Neumann boundary
  conditions.}
\label{FigExp3.2}
\end{figure}

\begin{figure}[htbp]
\centering
\begin{tabular}{ccc}
\includegraphics[scale=0.3]{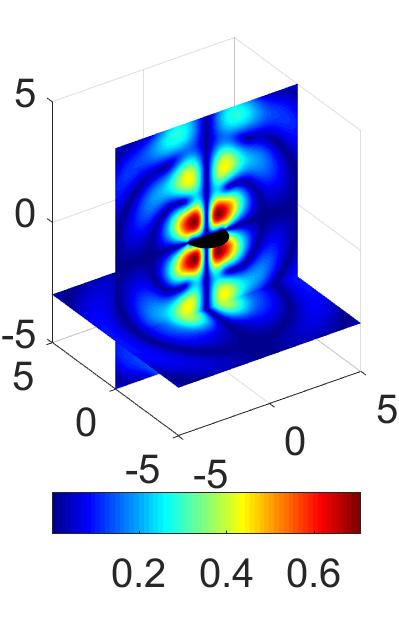} &
\includegraphics[scale=0.3]{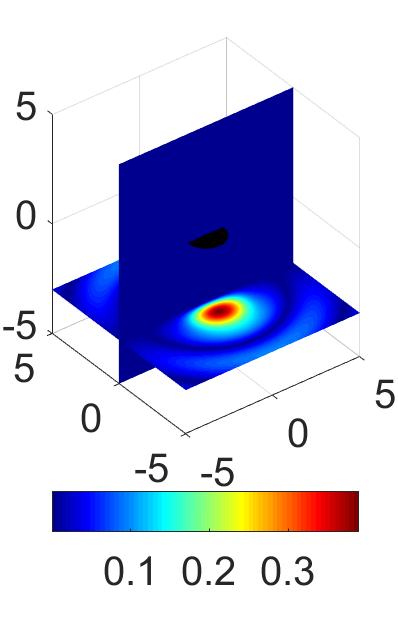} &
\includegraphics[scale=0.3]{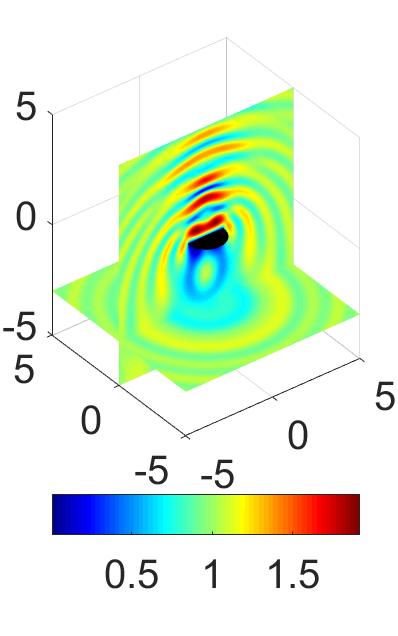} \\
(a) $|u^{\mathrm{num}}_1|$ & (b) $|u^{\mathrm{num}}_2|$ & (c) $|u^{\mathrm{num}}_3|$ \\
\end{tabular}
\caption{Scattering of a plane-wave pressure incident wave by a unit
  disc of diameter $2\lambda_s$. Total field under Dirichlet boundary
  conditions.}
\label{FigExp3.3}
\end{figure}

\section{Conclusions}
\label{sec:6}

This paper introduced novel regularized integral formulations and
associated fast high-order algorithms for the solution of 3D elastic
scattering problems with Neumann and Dirichlet boundary conditions on
closed and open surfaces. It was shown that the rectangular
integration method~\cite{BG18} and associated Chebyshev
differentiation strategies reliably provide high-order accuracies for
the weakly singular, strongly singular and hypersingular operators
associated with the closed- and open-surface formulations
considered. Relying on the newly studied Calder\'on formulation for 3D
elastic waves, the new integral operators inherent in the regularized
integral formulations enjoy excellent spectral properties and can give
rise to significantly reduced GMRES iterations numbers for a given
GMRES tolerance. For the problems with Dirichlet boundary conditions
on open surfaces, in turn, application of the weighted single-layer
operator is preferable. The regularized integral equation methods for
other scattering problems (for example, elastic transmission problems,
thermo- and porous-elastic problems, open-surface electromagnetic
problems) are left for future work.


\begin{thebibliography}{00}
\bibitem{AK02} C. Alves, R. Kress, On the far-field operator in elastic obstacle scattering, IMA J. Appl. Math. 67 (2002) 1-21.
\bibitem{AD95} C. Alves, T.H. Duong, Numerical resolution of the boundary integral equations for elastic scattering by a plane crack, Int. J. Numer. Meth. Eng. 38 (1995) 2347-2371.
\bibitem{AJKKY} K. Ando, Y. Ji, H. Kang, K. Kim, S. Yu, Spectral properties of the Neumann-Poincar\'e operator and cloaking by anomalous localized resonance for the elasto-static system, Euro. J. Appl. Math 29 (2018) 189-225.
\bibitem{AKM} K. Ando, H. Kang, Y. Miyanishi, Elastic Neumann-Poincar\'e operators on three dimensional smooth domains: Polynomial compactness and spectral structure, Int. Math. Res. Notices 2019(12) (2019) 3883-3900.
\bibitem{AD05} X. Antoine, M. Darbas, Alternative integral equations for the iterative solution of acoustic scattering problems. Quarterly J. Mech. Appl. Math. 58(1) (2005) 107-128.
\bibitem{AD07} X. Antoine, M. Darbas, Generalized combined field integral equations for the iterative solution of the three-dimensional Helmholtz equation, In Mathematical modeling and numerical analysis 41 (2007) 147-167.
\bibitem{BHSY} G. Bao, G. Hu, J. Sun, T. Yin, Direct and inverse elastic scattering from anisotropic media, J. Math. Pures Appl. 117 (2018) 263-301.
\bibitem{BXY} G. Bao, L. Xu, T. Yin, An accurate boundary element method for the exterior elastic scattering problem in two dimensions, J. Comput. Phy. 348 (2017) 343-363.
\bibitem{BXY19} G. Bao, L. Xu, T. Yin, Boundary integral equation methods for the elastic and thermoelastic waves in three dimensions, Comput. Method Appl. Methanics Eng. 354 (2019) 464-486.
\bibitem{BT98} M. Benzi, M. Tuma, A sparse approximate inverse preconditioner for nonsymmetric linear systems, SIAM J. Sci. Comput. 3(19) (1998) 968-994.
\bibitem{BET12} O.P. Bruno, T. Elling, C. Turc, Regularized integral equations and fast high-order solvers for sound-hard acoustic scattering problems, Int. J. Numer. Meth. Eng. 91 (2012) 1045-1072.
\bibitem{BG18} O.P. Bruno, E. Garza, A Chebyshev-based rectangular-polar integral solver for scattering by general geometries described by non-overlapping patches, available at arXiv:1807.01813.
\bibitem{BL12} O.P. Bruno, S. Lintner, Second-kind integral solvers for TE and TM problems of diffraction by open arcs, Radio Sci. 47 (6) (2012).
\bibitem{BL13} O.P. Bruno, S. Lintner, A high-order integral solver for scalar problems of diffraction by screens and apertures in three-dimensional space, J. Comput. Phy. 252 (2013) 250--274.
\bibitem{BK01} O.P. Bruno, L. Kunyansky, A fast, high-order algorithm for the solution of surface scattering problems: Basic implementation, tests, and applications, J. Comput. Phys. 169 (1) (2001) 80-110.
\bibitem{BXY191} O.P. Bruno, L. Xu, T. Yin, Weighted integral solvers for elastic scattering by open arcs in two dimensions, available at arxiv:1902.08687.
\bibitem{BLR14} F. Bu, J. Lin, F. Reitich, A fast and high-order method for the three-dimensional elastic wave scattering problem, J. Comput. Phy. 258 (2014) 856-870.
\bibitem{CDGS05} B. Carpentieri, I. Duff, L. Giraud, G. Sylvand, Combining fast multipoles techniques and an approximate inverse preconditioner for large electromagnetism calculations, SIAM J. Sci. Comput. 27(3) (2005) 774-792.
\bibitem{CBS08} S. Chaillat, M. Bonnet, J.-F. Semblat, A multi-level fast multipole BEM for 3-d elastodynamics in the frequency domain, Comput. Methods Appl. Mech. Eng. 197 (2008) 4233-4249.
\bibitem{CKM00} R. Chapko, R. Kress, L. Monch, on the numerical solution of a hypersingular integral equation for elastic scattering from a planar crack, IMA J. Numer. Anal. 20(4) (2000) 345-360.
\bibitem{CK98} D. Colton and R. Kress, Inverse Acoustic and Electromagnetic Scattering Theory, Berlin, Springer, 1998.
\bibitem{CDD03} M. Costabel, M. Dauge, R. Duduchava, Asymptotics without logarithmic terms for crack problems, Commun. Partial Differ. Equ. 28 (2003) 869-926.
\bibitem{DL15} M. Darbas, F. Le Lou\"er, Well-conditioned boundary integral formulations for high-frequency elastic scattering problems in three dimensions, Math. Meth. Appl. Sci. 38 (2015) 1705-1733.
\bibitem{GN78} J. Giroire, J. C. N\'{e}d\'{e}lec, Numerical solution of an exterior Neumann problem using a double layer potential, Math. Comp. 32 (1978) 973-990.
\bibitem{G72} M. E. Gurtin, The Linear Theory of Elasticity, Handbuch der Physik v. VIa/2, Springer-Verlag, New York-Heidelberg-Berlin, 1972.
\bibitem{HW08} G. C. Hsiao, W. L. Wendland, Boundary Integral Equations, Applied Mathematical Sciences, Vol. 164, Springer-verlag, 2008.
\bibitem{KGBB79} V. D. Kupradze, T. G. Gegelia, M. O. Basheleishvili, T. V. Burchuladze, Three-Dimensional Problems of the Mathematical Theory of Elasticity and Thermoelasticity, North-Holland Series in Applied Mathematics and Mechanics, vol. 25, North-Holland Publishing Co., Amsterdam, 1979.
\bibitem{LB15} S. Lintner, O. Bruno, A generalized Calder\'on formula for open-arc diffraction problems: Theoretical considerations, Proceedings of the Royal Society of Edinburgh 145A (2015) 331-364.
\bibitem{L09} Y. Liu, Fast Multipole Boundary Element Method, Cambridge University Press, New York, 2009.
\bibitem{LR93} Y. Liu, F. J. Rizzo, Hypersingular boundary integral equations for radiation and scattering of elastic waves in three dimensions, Comput. Method Appl. Method Eng. 107 (1993) 131-144.
\bibitem{L14} F. Le Lou\"er, A high order spectral algorithm for elastic obstacle scattering in three dimensions, J. Comput. Phy. 279 (2014) 1-18.
\bibitem{MB88} G. D. Manolis, D. E. Beskos, Boundary element methods in elastodynamics, Unwin Hyman, London, 1988.
\bibitem{N01} J. C. N\'{e}d\'{e}lec, Acoustic and Electromagnetic Equations: Integral Representations for Harmonic Problems, Springer-Verlag, New York, 2001.
\bibitem{TC07} M. S. Tong, W. C. Chew, Nystr\"om method for elastic wave scattering by three-dimensional obstacles, J. Comput. Phy. 226 (2007) 1845-1858.
\bibitem{TC09} M. S. Tong, W. C. Chew, Multilevel fast multipole algorithm for elastic wave scattering by large three-dimensional objects, J. Comput. Phy. 228 (2009) 921-932.
\bibitem{YHX} T. Yin, G. C. Hsiao, L. Xu, Boundary integral equation methods for the two dimensional fluid-solid interaction problem, SIAM J. Numer. Anal. 55(5) (2017) 2361-2393.


\end{thebibliography}
\end{document}